\documentclass[preprint,12pt]{elsarticle}

\usepackage{amsmath}
\usepackage{amssymb}

\newenvironment{proof}{\noindent{\bf Proof:}}
{\begin{flushright} \vspace{-0.5cm} $\square$ \end{flushright}}

\newcommand {\abs}[1]{\left\vert#1\right\vert}
\newcommand {\set}[1]{\left\{#1\right\}}

\newcommand {\defined} {\stackrel{def} {=}}

\newtheorem{coro}     {Corollary}       [section]
\newtheorem{claim}    {Claim}           [section]
\newtheorem{observation} {Observation}  [section]

\newtheorem{theorem}  {Theorem}         [section]
\newtheorem{lemma}    {Lemma}           [section]

\usepackage{algorithm}
\usepackage{algorithmicx}
\usepackage{algpseudocode}
\usepackage{algpascal}
\usepackage{graphicx}
\usepackage{subfigure}
\usepackage[usenames,dvipsnames]{pstricks}
\usepackage{graphics}
\usepackage{epsfig}
\usepackage{subfig}
\usepackage{booktabs}
\usepackage{multirow}
\usepackage{paralist}

\usepackage{lineno}

\newcommand{\runningtitle}[1]{\vspace{0.5ex}\noindent{{\small \textbf{\boldmath #1~}}}\\}

\newcommand{\eptg}[1] {\textsc{Ept}(#1)}
\newcommand{\vptg}[1] {\textsc{Vpt}(#1)}
\newcommand{\enptg}[1] {\textsc{Enpt}(#1)}
\newcommand{\ept} {\textsc{EPT}}
\newcommand{\vpt} {\textsc{VPT}}
\newcommand{\enpt} {\textsc{ENPT}}
\newcommand{\pp} {{\cal P}}
\newcommand{\eptgp} {\eptg{\pp}}
\newcommand{\vptgp} {\vptg{\pp}}
\newcommand{\enptgp} {\enptg{\pp}}
\newcommand{\prb} {\textsc{HamiltonianPairRec}}
\newcommand{\prbpthree} {\textsc{P3-HamiltonianPairRec}}

\newcommand{\reptp}[2] {\left< #1,#2 \right>}
\newcommand{\rep} {\reptp{T}{\pp}}
\newcommand{\repn}[1] {\reptp{T_{#1}}{\pp_{#1}}}
\newcommand{\ppbar} {\bar{\pp}}
\newcommand{\repbar} {\reptp{\bar{T}}{\ppbar}}

\renewcommand{\split} {\textit{split}}
\newcommand{\op} {{\cal O}}
\newcommand{\opg} {\op(G,C)}

\newcommand{\wdt} {{\cal W}}
\newcommand{\wdtg} {\wdt(G,C)}
\newcommand{\wdtgprime} {\wdt(G',C')}
\newcommand{\wdtgdprime} {\wdt(G'',C'')}

\newcommand{\contract}[2] {{#1}_{/{#2}}}
\newcommand{\contractge} {\contract{G}{e}}
\newcommand{\contractp}[1] {\contract{\rep}{#1}}
\newcommand{\contractppq}{\contractp{P_p,P_q}}
\newcommand{\contractgce} {\contract{(G,C)}{e}}
\journal{Discrete Applied Mathematics}

\begin{document}
\begin{frontmatter}
\title{Graphs of Edge-Intersecting Non-Splitting Paths in a Tree: Representations of Holes - Part I\tnoteref{wg}\tnoteref{ISF-TUBITAK}}
\tnotetext[ictcs]{A preliminary version of this paper appeared in \cite{BESZ13}.}
\tnotetext[ISF-TUBITAK]{This work was supported in part by the Israel Science Foundation grant No. 1249/08, by TUBITAK PIA BOSPHORUS Grant No. 111M303, by the TUBITAK 2221 Programme, Bogazici University Scientific Research Fund  BAP grant no 6461D and by the Catedra de Excelencia of Universidad Carlos III de Madrid, Departmento de Ingenieria Telematica.}

\author[bu]{Arman Boyac{\i}}
\ead{arman.boyaci@boun.edu.tr}

\author[bu]{T{\i}naz Ekim}
\ead{tinaz.ekim@boun.edu.tr}

\author[bu,telhai]{Mordechai Shalom\corref{corr}}
\ead{cmshalom@telhai.ac.il}
\cortext[corr]{Corresponding Author}

\author[technion]{Shmuel Zaks}
\ead{zaks@cs.technion.ac.il}

\address[bu]{Department of Industrial Engineering, Bogazici University, Istanbul, Turkey}
\address[telhai]{TelHai Academic College, Upper Galilee, 12210, Israel}
\address[technion]{Department of Computer Science, Technion, Haifa, Israel}

\begin{abstract}
Given a tree and a set $\pp$ of non-trivial simple paths on it, $\vptgp$ is the $\vpt$ graph (i.e. the vertex intersection graph) of the paths $\pp$ of the tree $T$, and $\eptgp$ is the $\ept$ graph (i.e. the edge intersection graph) of $\pp$. These graphs have been extensively studied in the literature.
Given two (edge) intersecting paths in a graph, their \emph{split vertices} is the set of vertices having degree at least $3$ in their union. A pair of (edge) intersecting paths is termed \emph{non-splitting} if they do not have split vertices (namely if their union is a path).

In this work, motivated by an application in all-optical networks, we define the graph $\enptgp$ of edge-intersecting non-splitting paths of a tree, termed the $\enpt$ graph, as the (edge) graph having a vertex for each path in $\pp$, and an edge between every pair of paths that are both edge-intersecting and non-splitting. A graph $G$ is an $\enpt$ graph if there is a tree $T$ and a set of paths $\pp$ of $T$ such that $G=\enptgp$, and we say that $\rep$ is a \emph{representation} of $G$. We first show that cycles, trees and complete graphs are $\enpt$ graphs.

Our work follows the lines of Golumbic and Jamison's research \cite{Golumbic1985151,GJ85} in which they defined the $\ept$ graph class, and characterized the representations of chordless cycles (holes). It turns out that $\enpt$ holes have a more complex structure than $\ept$ holes. In our analysis, we assume that the $\ept$ graph corresponding to a representation of an $\enpt$ hole is given. We also introduce three assumptions $(P1)$, $(P2)$, $(P3)$ defined on $\ept$, $\enpt$ pairs of graphs. In this Part I,  using the results of Golumbic and Jamison as building blocks, we characterize (a) $\ept$, $\enpt$ pairs that satisfy $(P1)$, $(P2)$, $(P3)$, and (b) the unique minimal representation of such pairs.


\end{abstract}

\begin{keyword}
Intersection Graphs \sep EPT Graphs
\end{keyword}
\end{frontmatter}

\section{Introduction}\label{sec:intro}
Given a tree $T$ and a set $\pp$ of non-trivial simple paths in $T$, the Vertex Intersection Graph of Paths in a Tree (\vpt) and  the Edge Intersection Graph of Paths in a Tree ($\ept$) of $\pp$ are denoted by $\vptgp$ and $\eptgp$, respectively. Both graphs have $\pp$ as vertex set. $\vptgp$ (resp. $\eptgp$) contains an edge between two vertices if the corresponding two paths intersect in at least one vertex (resp. edge). A graph $G$ is $\vpt$ (resp. $\ept$) if there exist a tree $T$ and a set $\pp$ of non-trivial simple paths in $T$ such that $G$ is isomorphic to $\vptgp$ (resp. $\eptgp$). In this case we say that $\rep$ is a $\vpt$ (resp. an $\ept$) representation of $G$.

In this work we focus on edge intersections of paths. The graph of edge-intersecting and non-splitting paths of a tree ($\enpt$) of a given representation $\rep$ denoted by $\enptgp$, has a vertex $v$ for each path $P_v$ of $\pp$ and two vertices $u,v$ of $\enptgp$ are adjacent if the paths $P_u$ and $P_v$ edge-intersect and do not split (that is, their union is a path). A graph $G$ is an $\enpt$ graph if there is a tree $T$ and a set of paths $\pp$ of $T$ such that $G$ is isomorphic to $\enptgp$. We note that $\eptgp=\enptgp$ is an interval graph whenever $T$ is a path. Therefore the class $\enpt$ includes all interval graphs.

\subsection{Applications}
$\ept$ and $\vpt$ graphs have applications in communication networks. Consider a communication network of a tree topology $T$. The message routes to be delivered in this communication network are paths on $T$. Two paths conflict if they both require to use the same link (node). This conflict model is equivalent to an $\ept$ (a $\vpt$) graph. Suppose we try to find a schedule for the messages such that no two messages sharing a link (node) are scheduled in the same time interval. Then a vertex coloring of the $\ept$ ($\vpt$) graph corresponds to a feasible schedule on this network.

$\ept$ graphs also appear in all-optical telecommunication networks. The so-called Wavelength Division Multiplexing (WDM) technology can multiplex different signals onto a single optical fiber by using different wavelength ranges of the laser beam \cite{CGK92,R93}. WDM is a promising technology enabling us to deal with the massive growth of traffic in telecommunication networks, due to applications such as video-conferencing, cloud computing and distributed computing \cite{DV93}. A stream of signals traveling from its source to its destination in optical form is termed a \emph{lightpath}. A lightpath is realized by signals traveling through a series of fibers, on a certain wavelength. Specifically, Wavelength Assignment problems (WLA) are a family of path coloring problems that aim to assign wavelengths (i.e. colors) to lightpaths, so that no two lightpaths with a common edge receive the same wavelength and certain objective function (depending on the problem) is minimized.

\emph{Traffic Grooming} is the term used for the combination of several low (i.e. sub-wavelength) capacity requests (modeled by paths of a network) into one lightpath (modeled by a path or cycle of the network) using Time Division Multiplexing (TDM) technology \cite{GRS98}. In this context a set of  paths can be combined into one lightpath, as long as they satisfy the following two conditions:
\begin{itemize}
\item The load condition: On any given fiber, at most $g$ requests can use the same lightpath, where $g$ is an integer termed \emph{the grooming factor}.
\item The no-split condition: a lightpath (i.e. the union of the requests using the lightpath) constitutes a path or a cycle of the network.
\end{itemize}
Clearly, the second condition cannot be tested in the $\ept$ model. For this reason we introduce the $\enpt$ graphs that provide the required information.

Readers unfamiliar with optical networks may consider the following analogous problem in transportation. Consider a set of transportation requests modeled by paths, and trucks traveling along paths or cycles. Trucks are able to load and drop items during their journey as long as at any given time their load does not exceed their capacity. The no-split condition reflects the fact that a truck has to follow a path or a cycle.

By the no-split condition, a (feasible) traffic grooming corresponds to a vertex coloring of the graph $(V(\eptgp), E(\eptgp) \setminus E(\enptgp))$. Moreover, by the load condition, every color class induces a sub-graph of $\eptgp$ with clique number at most $g$. Therefore, it makes sense to analyze the structure of these graph pairs, i.e. the two graphs $\eptgp$ and $\enptgp$ defined on the same set of vertices.

Under this setting one can consider various objective functions such as:
\begin{itemize}
\item{Minimize the number of wavelengths / trucks:}
When the number of wavelengths (resp. trucks) is scarce, one aims to minimize this number. We note that
when the parameter $g$ is sufficiently big (i.e. $g=\infty$) the problems boils down to the minimum vertex coloring problem.

\item{Minimize the number of regenerators / total distance traveled:}
The signal traveling on a lightpath has to be regenerated along its way, implying a regeneration cost roughly proportional to its length. Similarly, a truck incurs operation expenses proportional to the distance it travels.

In order to model this problem, one has to assign weights to the $\enpt$ edges indicating the length of the overlap of two non-splitting requests.
The gains obtained from putting two requests in the same lightpath (truck) is exactly the weight of the corresponding $\enpt$ edge. Therefore, an optimal solution corresponds to a partition of the vertices into paths of $\enptgp$ (without a chord from $\eptgp$ and with a maximum clique size of at most $g$), such that the sum of the weights of the edges of these paths is maximized.
\end{itemize}

\subsection{Related Work}
$\ept$ and $\vpt$ graphs have been extensively studied in the literature. Although $\vpt$ graphs can be characterized by a fixed number of forbidden subgraphs \cite{Leveque2009Characterizing}, it is shown that $\ept$ graph recognition is NP-complete \cite{Golumbic1985151}. Edge intersection and vertex intersection give rise to identical graph classes in the case of paths in a tree having maximum degree 3 \cite{Golumbic1985151}. However, $\vpt$ graphs and $\ept$ graphs are incomparable in general; neither $\vpt$ nor $\ept$ contains the other. Main optimization and decision problems such as recognition \cite{Fanica1978211}, the maximum clique \cite{Fanica2000181}, the minimum vertex coloring \cite{Golumbic:2004:AGT:984029} and the maximum stable set problems \cite{Spinrad1995181} are polynomial-time solvable in $\vpt$
whereas the minimum vertex coloring problems remain NP-complete in $\ept$ graphs \cite{Golumbic1985151,GJ85}. In contrast, one can solve in polynomial time the maximum clique \cite{GJ85} and the maximum stable set \cite{RobertE1985221} problems in $\ept$ graphs. In addition, \cite{GJ85} shows that holes have a unique $\ept$ representation, called a \emph{pie}. Forbidden subgraphs of $\ept$ graphs implied by this result are also provided.

After these works on $\ept$ and $\vpt$ graphs in the early 80's, this topic did not get focus until very recently. Current research on intersection graphs is concentrated on the comparison of various intersection graphs of paths in a tree and their relation to chordal and weakly chordal graphs \cite{Golumbic2008Equivalences,GolumbicLS08}. Also, a tolerance model is studied via $k$-edge intersection graphs where two vertices are adjacent if their corresponding paths intersect on at least $k$ edges \cite{Golumbic2008Kedge}. Besides, several recent papers consider the edge intersection graphs of paths on a grid (e.g. \cite{BiedlS10,GolumbicLS09}).

\subsection{Our Contribution}
In this work we define the new family of $\enpt$ graphs, and investigate its basic properties. To this aim, we first study possible $\enpt$ representations of some basic structures such as trees, cliques and holes.

It turns out that in $\enpt$ graphs, representations of chordless cycles have a much more complicated structure, yielding several possible representations, see Figure \ref{fig:ept-vs-enpt-cycle-representations}. In fact, given a hole $C$, several $\enpt$ representations $\rep$ such that $\enptgp$ is isomorphic to $C$ but $\eptgp$ are non-isomorphic to each other are possible, see Figure \ref{fig:eptn-c4}.

\begin{figure}[htbp]
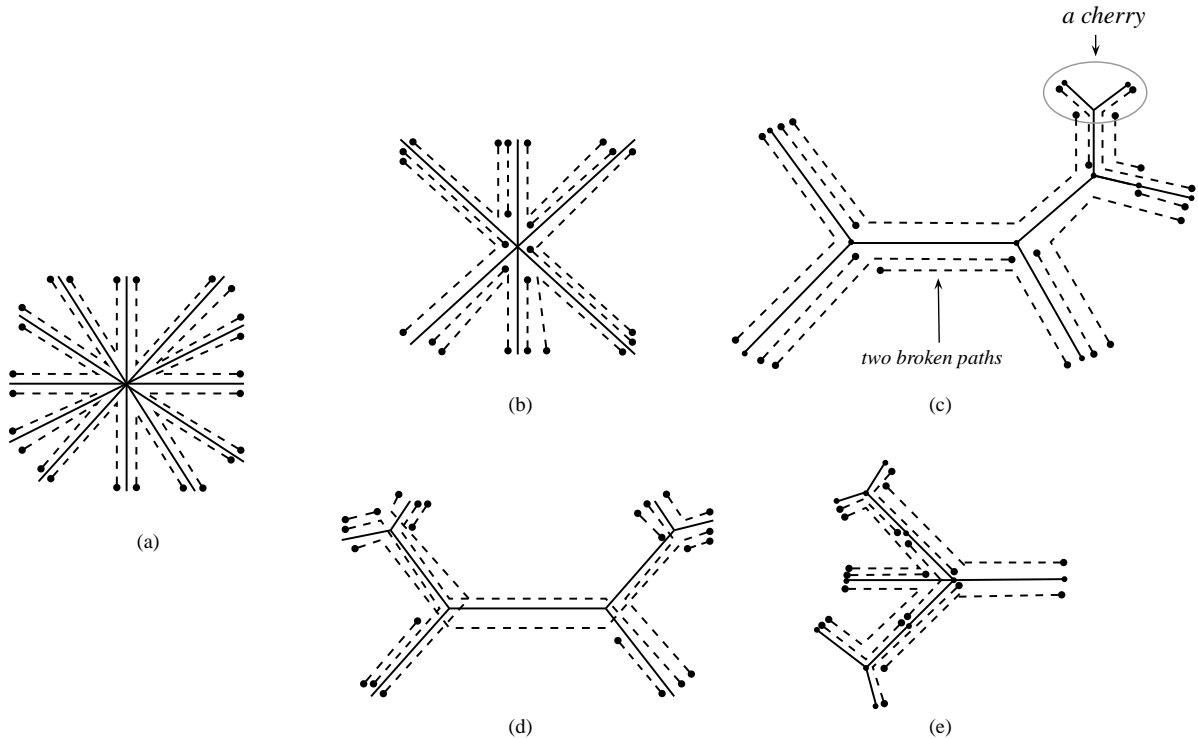

\centering
\include{figure/ept-vs-enpt-cycle-representations}
\caption{(a) The $\ept$ representation of a $C_{12}$, (b) a simple $\enpt$ representation of a $C_{12}$ (c) a broken planar tour with cherries representation of a $C_{12}$. In Part II of this work \cite{BESZ13-TR2}, we showed that all $\enpt$ representations satisfying the condition $(P3)$, defined in Section \ref{sec:EPTNBasics}, of a cycle have this form. (d) a non-planar tour representation of a $C_{12}$, (e) a non-tour representation of a $C_{10}$.}
\label{fig:ept-vs-enpt-cycle-representations}
\end{figure}

Consider the pair $(G,C)$ where $G$ is a graph and $C$ is a Hamiltonian cycle of $G$. We restrict our attention to the determination of a representation $\rep$ such that $\eptgp = G$ and $\enptgp = C$. In this case we will say that $\rep$ is a representation of $(G,C)$. We introduce three assumptions $(P1)$, $(P2)$, $(P3)$ defined on $\ept$, $\enpt$ pairs of graphs. In this Part I, using the results of Golumbic and Jamison as building blocks, we characterize (a) $\ept$, $\enpt$ pairs that satisfy $(P1)$, $(P2)$, $(P3)$, and (b) the unique minimal representation of such pairs.

In Part II of this work \cite{BESZ13-TR2} we relax these assumptions and show that the results can be extended to the case where only $(P3)$ holds. A family of non-$\enpt$ graphs is obtained as a result of this extension. For the general case (without assumption $(P3)$) we show that it is NP-complete to determine whether a given $\ept$, $\enpt$ pair has a representation even when the $\enpt$ graph is a cycle. This result extends the NP-completeness of $\ept$ graph recognition \cite{Golumbic1985151}.

This paper is organized as follows: in Section \ref{sec:prelim} we give basic definitions and preliminaries that we use in developing our results. We provide $\enpt$ representations of basic graphs such as trees, cliques and cycles, thus showing that they are included in the family of $\enpt$ graphs. We also characterize all the $\enpt$ representations of cliques. In Section \ref{sec:EPTNBasics} we obtain basic results regarding $\enpt$ graphs, their relationship with $\ept$ graphs, and their (common) representations. We then define the contraction operation, which is basically replacing two paths with their union provided that this union is a path. We also introduce the assumptions $(P1)$, $(P2)$ and $(P3)$, under which in Section \ref{sec:unionminimal}, we characterize the representations of $\enpt$ holes. More specifically, we characterize the representations $\rep$ of pairs $(G,C)$, where $C$ is a Hamiltonian cycle of $G$, such that $\eptgp=G$ and $\enptgp=C$ and satisfy these assumptions. Our work \cite{BESZ13-TR2} considers the relaxation of these assumptions.

\section{Preliminaries and Basic Results}\label{sec:prelim}
In this section we provide definitions used in the paper, present known results related to our work, and develop basic results. The section is organized as follows: Section \ref{subsec:prelim-definitions} is devoted to basic definitions, in Section \ref{subsec:prelim-eptn-knownresults} we present known results on $\ept$ graphs that are closely related to our work and in Section \ref{subsec:prelim-some-enpt-graphs} we present some graph families that are contained in the family of $\enpt$ graphs.

\subsection{Definitions}\label{subsec:prelim-definitions}
\runningtitle{General Notation:}
Given a graph $G$ and a vertex $v$ of $G$, we denote by $d_G(v)$ the degree of $v$ in $G$. A vertex is called a \emph{leaf} (resp. \emph{intermediate vertex}, \emph{junction}) if $d_G(v)=1$ (resp. $=2$, $\geq 3$). Whenever there is no ambiguity we omit the subscript $G$ and write $d(v)$.

Given a graph $G$, $\bar{V} \subseteq V(G)$ and $\bar{E} \subseteq E(G)$ we denote by $G[\bar{V}]$ and $G[\bar{E}]$ the subgraphs of $G$ induced by $\bar{V}$ and by $\bar{E}$, respectively.

The \emph{union} of two graphs $G, G'$ is the graph $G \cup G' \defined (V(G) \cup V(G'), E(G) \cup E(G'))$. The \emph{join} $G+G'$ of two disjoint graphs $G,G'$ is the graph $G \cup G'$ together with all the edges joining $V(G)$ and $V(G')$, i.e. $G + G' \defined (V(G) \cup V(G'), E(G) \cup E(G') \cup (V(G) \times V(G')))$.

Given a (simple) graph $G$ and $e \in E(G)$, we denote by $\contractge$ the (simple) graph obtained by contracting the edge $e = \left\{ p, q \right\}$ of $G$, i.e. by coinciding the two endpoints of $e$ to a single vertex $p.q$ and removing self loops and parallel edges.

For any two vertices $u,v$ of a tree $T$, we denote by $p_T(u,v)$ the unique path between $u$ and $v$ in $T$. We denote the set of all positive integers at most $k$ as $[k]$.

\runningtitle{Intersections and union of paths:}
Given two paths $P,P'$ in a graph, we write $P \parallel P'$ to denote that $P$ and $P'$ are \emph{edge-disjoint}. The \emph{split vertices} of $P$ and $P'$ is the set of junctions in their union $P \cup P'$ and is denoted by $\split(P,P')$. Whenever $P$ and $P'$ edge-intersect and $\split(P, P') = \emptyset$  we say that $P$ and $P'$ are \emph{non-splitting} and denote this by $P \sim P'$. In this case $P \cup P'$ is a path or a cycle. When $P$ and $P'$ edge-intersect and $\split(P, P') \neq \emptyset$ we say that they are \emph{splitting} and denote this by $P \nsim P'$.  Clearly, for any two paths $P$ and $P'$ exactly one of the following holds: $P \parallel P'$, $P \sim P'$, $P \nsim P'$.

When the graph $G$ is a tree, the union $P \cup P'$ of two edge-intersecting paths $P,P'$ of $G$ is a tree with at most two junctions, i.e. $\abs{\split(P,P')} \leq 2$ and $P \cup P'$ is a path whenever $P \sim P'$.

\runningtitle{The VPT, EPT and ENPT graphs:}
Let $\pp$ be a set of paths in a tree $T$. The graphs $\vptgp, \eptgp$ and $\enptgp$ are graphs such that $V(\enptgp) = V(\eptgp) = V(\vptgp)=\set{p | P_p \in \pp)}$. Given two distinct paths $P_p,P_q \in \pp$, $\set{p,q}$ is an edge of $\enptgp$ if $P_p \sim P_q$, and $\set{p,q}$ is an edge of $\eptgp$ (resp. $\vptgp$) if $P_p$ and $P_q$ have a common edge (resp. vertex) in $T$. See Figure \ref{fig:vpt-ept-eptn} for an example. From these definitions it follows that:
\begin{observation}
$E(\enptgp) \subseteq E(\eptgp)\subseteq E(\vptgp)$.
\end{observation}
Two graphs $G$ and $G'$ such that $V(G)=V(G')$ and $E(G') \subseteq E(G)$ are termed a \emph{pair (of graphs)} denoted as $(G,G')$. If $\eptgp=G$ (resp. $\enptgp=G$) we say that $\rep$ is an $\ept$ (resp. $\enpt$) representation for $G$. If $\eptgp=G$ and $\enptgp=G'$ we say that $\rep$ is a representation for the pair $(G,G')$. Given a pair $(G,G')$ the sub-pair induced by $\bar{V} \subseteq V(G)$ is the pair $(G[\bar{V}],G'[\bar{V}])$. Clearly, any representation of a pair induces representations for its induced sub-pairs, i.e. the pairs have the hereditary property.

Throughout this work, in all figures, the edges of the tree $T$ of a representation $\rep$ are drawn as solid edges whereas the paths on the tree are shown by dashed, dotted, etc. edges. Similarly, edges of $\enptgp$ are drawn with solid or blue lines whereas edges in $E(\eptgp) \setminus E(\enptgp)$ are dashed or red. We sometimes refer to them as blue and red edges, respectively. For an edge $e=\set{p,q}$ we use $\split(e)$ as a shorthand for $\split(P_p,P_q)$. We note that $e$ is a red edge if and only if $\split(e) \neq \emptyset$.

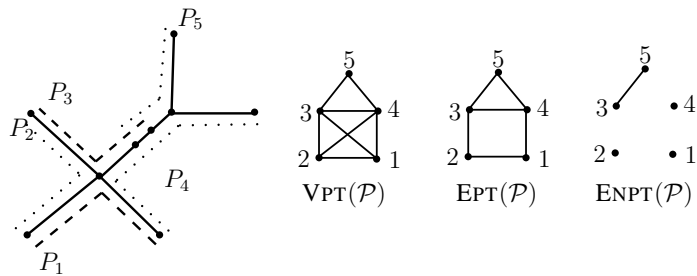
\begin{figure}[htbp]
\centering
\scalebox{0.8} 
{
\begin{pspicture}(0,-2.238125)(11.862812,2.238125)
\psline[linewidth=0.04cm](0.5609375,0.4996875)(1.7209375,-0.5603125)
\psline[linewidth=0.04cm](1.7209375,-0.5603125)(0.5209375,-1.5403125)
\psline[linewidth=0.04cm](1.7209375,-0.6003125)(2.9209375,0.4996875)
\psline[linewidth=0.04cm](1.7409375,-0.5803125)(2.7409375,-1.5603125)
\psline[linewidth=0.04cm](2.9209375,0.5196875)(2.9609375,1.7996875)
\psline[linewidth=0.04cm](2.9409375,0.4796875)(4.3009377,0.4796875)
\psline[linewidth=0.04cm,linestyle=dashed,dash=0.16cm 0.16cm](0.7209375,0.5596875)(1.6609375,-0.3203125)
\psline[linewidth=0.04cm,linestyle=dashed,dash=0.16cm 0.16cm](1.6609375,-0.3203125)(2.4609375,0.3596875)
\psline[linewidth=0.04cm,linestyle=dotted,dotsep=0.16cm](2.0009375,-0.6403125)(3.0409374,0.2796875)
\psline[linewidth=0.04cm,linestyle=dotted,dotsep=0.16cm](3.0409374,0.2796875)(4.3609376,0.2996875)
\psline[linewidth=0.04cm,linestyle=dotted,dotsep=0.16cm](1.9809375,-0.6403125)(2.8809376,-1.4803125)
\psline[linewidth=0.04cm,linestyle=dotted,dotsep=0.16cm](2.0809374,0.2596875)(2.7209375,0.7196875)
\psline[linewidth=0.04cm,linestyle=dotted,dotsep=0.16cm](2.7209375,0.7196875)(2.8209374,1.8996875)
\psline[linewidth=0.04cm,linestyle=dotted,dotsep=0.16cm](0.5009375,0.2996875)(1.3809375,-0.5203125)
\psline[linewidth=0.04cm,linestyle=dotted,dotsep=0.16cm](1.4009376,-0.5403125)(0.3209375,-1.4403125)
\psline[linewidth=0.04cm,linestyle=dashed,dash=0.16cm 0.16cm](0.6609375,-1.7403125)(1.7609375,-0.8403125)
\psline[linewidth=0.04cm,linestyle=dashed,dash=0.16cm 0.16cm](1.7609375,-0.8403125)(2.6609375,-1.7003125)
\usefont{T1}{ptm}{m}{n}
\rput(0.9323437,-2.0103126){$P_1$}
\usefont{T1}{ptm}{m}{n}
\rput(0.43234375,0.2096875){$P_2$}
\usefont{T1}{ptm}{m}{n}
\rput(1.0723437,0.8096875){$P_3$}
\usefont{T1}{ptm}{m}{n}
\rput(3.0123436,-0.6303125){$P_4$}
\usefont{T1}{ptm}{m}{n}
\rput(3.2323437,2.0496874){$P_5$}
\psdots[dotsize=0.12](5.3709373,-0.2603125)
\psdots[dotsize=0.12](6.3309374,-0.2803125)
\psdots[dotsize=0.12](6.3509374,0.5196875)
\psdots[dotsize=0.12](5.8709373,1.1396875)
\psline[linewidth=0.03cm](5.3709373,-0.2603125)(6.3309374,-0.2603125)
\psdots[dotsize=0.12](5.3909373,0.5196875)
\psline[linewidth=0.03cm](6.3509374,0.5196875)(5.3909373,0.5196875)
\psline[linewidth=0.03cm](5.3709373,0.5196875)(5.3709373,-0.2203125)
\psline[linewidth=0.03cm](6.3309374,0.5196875)(6.3309374,-0.3003125)
\psline[linewidth=0.03cm](5.3909373,0.5396875)(5.8509374,1.1196876)
\psline[linewidth=0.03cm](5.8509374,1.1196876)(6.3509374,0.5196875)
\psline[linewidth=0.03cm](5.3709373,0.5196875)(6.3409376,-0.2803125)
\psline[linewidth=0.03cm](5.3709373,-0.2403125)(6.3509374,0.5196875)
\usefont{T1}{ptm}{m}{n}
\rput(6.612344,0.5496875){$4$}
\usefont{T1}{ptm}{m}{n}
\rput(5.1523438,0.4896875){$3$}
\usefont{T1}{ptm}{m}{n}
\rput(5.112344,-0.2703125){$2$}
\usefont{T1}{ptm}{m}{n}
\rput(6.632344,-0.2903125){$1$}
\usefont{T1}{ptm}{m}{n}
\rput(5.8723435,1.3696876){$5$}
\usefont{T1}{ptm}{m}{n}
\rput(5.7857814,-0.8748438){$\vptgp$}
\usefont{T1}{ptm}{m}{n}
\rput(8.323594,-0.8748438){$\eptgp$}
\usefont{T1}{ptm}{m}{n}
\rput(10.771563,-0.8748438){$\enptgp$}
\psdots[dotsize=0.12](7.8509374,-0.2403125)
\psdots[dotsize=0.12](8.810938,-0.2603125)
\psdots[dotsize=0.12](8.830937,0.5396875)
\psdots[dotsize=0.12](8.350938,1.1596875)
\psline[linewidth=0.03cm](7.8509374,-0.2403125)(8.810938,-0.2403125)
\psdots[dotsize=0.12](7.8709373,0.5396875)
\psline[linewidth=0.03cm](8.830937,0.5396875)(7.8709373,0.5396875)
\psline[linewidth=0.03cm](7.8509374,0.5396875)(7.8509374,-0.2003125)
\psline[linewidth=0.03cm](8.810938,0.5396875)(8.810938,-0.2803125)
\psline[linewidth=0.03cm](7.8709373,0.5596875)(8.330937,1.1396875)
\psline[linewidth=0.03cm](8.330937,1.1396875)(8.830937,0.5396875)
\usefont{T1}{ptm}{m}{n}
\rput(9.092343,0.5696875){$4$}
\usefont{T1}{ptm}{m}{n}
\rput(7.632344,0.5096875){$3$}
\usefont{T1}{ptm}{m}{n}
\rput(7.592344,-0.2503125){$2$}
\usefont{T1}{ptm}{m}{n}
\rput(9.112344,-0.2703125){$1$}
\usefont{T1}{ptm}{m}{n}
\rput(8.352344,1.3896875){$5$}
\psdots[dotsize=0.12](10.290937,-0.1803125)
\psdots[dotsize=0.12](11.250937,-0.2003125)
\psdots[dotsize=0.12](11.270938,0.5996875)
\psdots[dotsize=0.12](10.790937,1.2196875)
\psdots[dotsize=0.12](10.310938,0.5996875)
\psline[linewidth=0.03cm](10.310938,0.6196875)(10.770938,1.1996875)
\usefont{T1}{ptm}{m}{n}
\rput(11.532344,0.6296875){$4$}
\usefont{T1}{ptm}{m}{n}
\rput(10.072344,0.5696875){$3$}
\usefont{T1}{ptm}{m}{n}
\rput(10.032344,-0.1903125){$2$}
\usefont{T1}{ptm}{m}{n}
\rput(11.552343,-0.2103125){$1$}
\usefont{T1}{ptm}{m}{n}
\rput(10.792344,1.4496875){$5$}
\psdots[dotsize=0.12](1.7209375,-0.5603125)
\psdots[dotsize=0.12](2.9209375,0.4996875)
\psdots[dotsize=0.12](2.9609375,1.7996875)
\psdots[dotsize=0.12](4.3009377,0.4996875)
\psdots[dotsize=0.12](0.5809375,0.4796875)
\psdots[dotsize=0.12](2.7209375,-1.5403125)
\psdots[dotsize=0.12](0.5209375,-1.5403125)
\psdots[dotsize=0.12](2.3209374,-0.0403125)
\psdots[dotsize=0.12](2.5809374,0.1996875)
\end{pspicture} 
}
\caption{A host tree $T$, a collection of paths $\pp = \left\{ P_1, P_2, P_3, P_4, P_5 \right\} $ defined on $T$ and the corresponding graphs $\vptgp, \eptgp$ and $\enptgp$.}
\label{fig:vpt-ept-eptn}
\end{figure}
\runningtitle{Cycles, Chords, Holes, Outerplanar graphs, Trees:}
Given a graph $G$ and a cycle $C$ of it, a \emph{chord} of $C$ in $G$ is an edge of $E(G) \setminus E(C)$ connecting two vertices of $V(C)$. The \emph{length} of a chord connecting the vertices $i$,$j$ is the length of a shortest path between $i$ and $j$ on $C$. $C$ is a \emph{hole} (chordless cycle) of $G$ if $G$ does not contain any chord of $C$. This is equivalent to saying that the subgraph $G[V(C)]$ of $G$ induced by the vertices of $C$ is a cycle. For this reason a chordless cycle is also called an \emph{induced} cycle.

An \emph{outerplanar} graph is a planar graph that can be embedded in the plane such that all its vertices are on the unbounded face of the embedding. An outerplanar graph is Hamiltonian if and only if it is biconnected; in this case the unbounded face forms the unique Hamiltonian cycle. The \emph{weak dual} graph of a planar graph $G$ is the graph obtained from its dual graph by removing the vertex corresponding to the unbounded face of $G$. The weak dual graph of an outerplanar graph is a forest, and in particular the weak dual graph of a Hamiltonian outerplanar graph is a tree~\cite{CH67}. When working with outerplanar graphs we use the term \emph{face} to mean a bounded face.

\subsection{Preliminaries on $\ept$ graphs} \label{subsec:prelim-eptn-knownresults}
We now present definitions and results from \cite{GJ85} that we use throughout this work.

A \emph{pie} of a representation $\rep$ of an $\ept$ graph is an induced star $K_{1,k}$ of $T$ with $k$ leaves $v_0, v_1, \ldots, v_{k-1} \in V(T)$, and $k$ paths $P_0, P_1, \ldots P_{k-1} \in \pp$, such that for every $0 \leq i \leq k-1$ both $v_i$ and $v_{(i+1)\mod k}$ are vertices of $P_i$. We term the central vertex of the star as the \emph{center} of the pie (see Figure \ref{fig:ept-cycle}). It is easy to see that the $\ept$ graph of a pie with $k$ leaves is the hole $C_k$ on $k$ vertices. Moreover, this is the only possible $\ept$ representation of $C_k$ when $k \geq 4$:

\begin{figure}
\centering
\scalebox{0.7} 
{
\begin{pspicture}(0,-2.05)(4.24,2.05)
\psline[linewidth=0.04cm](0.0,1.73)(1.98,0.11)
\psline[linewidth=0.04cm](1.96,0.11)(1.98,2.01)
\psline[linewidth=0.04cm](1.98,0.11)(3.9,1.75)
\psline[linewidth=0.04cm](1.98,0.11)(4.18,-0.65)
\psline[linewidth=0.04cm](1.96,0.13)(3.68,-2.03)
\psdots[dotsize=0.12](0.8,-0.01)
\psdots[dotsize=0.12](1.24,-0.63)
\psdots[dotsize=0.12](1.94,-0.99)
\psline[linewidth=0.04,linestyle=dashed,dash=0.16cm 0.16cm](0.12,1.81)(1.88,0.33)(1.86,2.03)
\psline[linewidth=0.04,linestyle=dashed,dash=0.16cm 0.16cm](2.06,2.03)(2.06,0.33)(3.78,1.81)
\psline[linewidth=0.04,linestyle=dashed,dash=0.16cm 0.16cm](3.96,1.65)(2.2,0.19)(4.22,-0.51)
\psline[linewidth=0.04,linestyle=dashed,dash=0.16cm 0.16cm](4.1,-0.73)(2.28,-0.11)(3.74,-1.91)
\usefont{T1}{ptm}{m}{n}
\rput(1.4214063,1.5){$P_{n-1}$}
\usefont{T1}{ptm}{m}{n}
\rput(2.6114063,1.26){$P_0$}
\usefont{T1}{ptm}{m}{n}
\rput(3.5714064,0.44){$P_1$}
\usefont{T1}{ptm}{m}{n}
\rput(3.5514061,-0.94){$P_2$}
\end{pspicture} 
}
\caption{The only $\ept$ representation of a cycle is a pie.}
\label{fig:ept-cycle}
\end{figure}
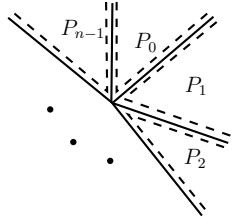

\begin{theorem}\cite{GJ85}\label{thm:golumbicpie}
If an $\ept$ graph contains a hole with $k \geq 4$ vertices, then every representation of it contains a pie with $k$ paths.
\end{theorem}

Let $\pp_e \defined \set{P \in \pp|~e \in P}$ be the set of paths in $\pp$ containing the edge $e$. A star $K_{1,3}$ is termed a \emph{claw}. For a claw $K$ of a tree $T$, $\pp[K] \defined \set{P \in \pp|~P \textrm{~uses two edges of~}K}$. It is easy to see that both $\eptg{\pp_e}$ and $\eptg{\pp[K]}$ are cliques. These cliques are termed \emph{edge-clique} and \emph{claw-clique}, respectively. Moreover, these are the only possible representations of cliques. We note that a claw-clique is a pie with $3$ leaves.

\begin{theorem}\cite{GJ85}\label{thm:golumbiccliques}
Any maximal clique of an $\ept$ graph with representation $\rep$ corresponds to a subcollection $\pp_e$ of paths for some edge $e$ of $T$, or to a subcollection $\pp[K]$ of paths for some claw $K$ of $T$.
\end{theorem}

\subsection{Some $\enpt$ graphs}\label{subsec:prelim-some-enpt-graphs}
In this section we show that trees, cycles and cliques are contained in the family of $\enpt$ graphs, and give a complete characterization of the $\enpt$ representations of cliques:

\begin{lemma}\label{lem:eptn-clique}
Every clique $K$ of $\enptgp$ corresponds to an edge-clique, such that the union of the paths representing $K$ is a path.
\end{lemma}
\begin{proof}
$\enptgp$ is a subgraph of $\eptgp$. Therefore a clique $K$ of $\enptgp$ is a clique of $\eptgp$. By Theorem \ref{thm:golumbiccliques}, $K$ corresponds to either an edge-clique or a claw-clique. Assume, by way of contradiction that $K$ does not correspond to an edge-clique. A claw-clique that is not an edge-clique, contains two paths $P_p, P_q$ each of which uses a different pair of the three edges of the claw. Therefore $P_p \nsim P_q$, i.e. $\set{p,q} \notin E(K)$, a contradiction. Therefore $K$ corresponds to an edge-clique. To show the second part of the claim, assume that the union of the paths corresponding to the vertices of $K$ is not a path. Then it contains at least one split vertex, i.e. it contains two paths $P_p,P_q$ such that $P_p \nsim P_q$, i.e. $\set{p,q} \notin E(K)$, a contradiction.
\end{proof}

A direct consequence of Lemma \ref{lem:eptn-clique} is that the maximum clique problem in $\enpt$ graphs can be solved in polynomial time. Let $G$ be an $\enpt$ graph and $\rep$ be an $\enpt$ representation for $G$. Consider an edge $e$ of $T$, the union of paths in $\pp_e$ induces a subtree $T_e$ of $T$. Let $l_1, l_2, \ldots, l_k \in V(T)$ be the leaves of $T_e$. Let $\pp_e^{l_i,l_j} \defined \set{P \in \pp_e|P \subseteq p_T(l_i,l_j) }$. The maximal cliques of $G$ correspond to the sets $\pp_e^{l_i,l_j}$. Therefore, there are at most $O(V(T)^3)$ maximal cliques in $G$. We conclude that (even if a representation $\rep$ is not given) a maximum clique can be found using a clique enumeration algorithm, e.g. \cite{TsukiyamaIdeAriyoshiShirakawa77}, since there are only a polynomial number of maximal cliques.

\begin{lemma}\label{lem:eptn-tree}
Every tree is an $\enpt$ graph.
\end{lemma}

\begin{proof}
Given a tree $T'$, the following procedure provides an $\enpt$ representation $\rep$ of $T'$: 1) $T \leftarrow T'$, 2) choose an arbitrary vertex $r$ as the root of $T$ and hang $T$ from $r$, 3) add two vertices $\bar{r}, \bar{\bar{r}}$ and two edges $\left\{ \bar{\bar{r}},\bar{r} \right\} \left\{ \bar{r},r \right\}$ to $T$, 4) $\pp=\set{P_v|~v \in T'}$ where $P_v$ is a path of length $2$ between $v$ and its ancestor at distance $2$. It remains to show that $\set{u,v} \in T'$ if and only if $P_u \sim P_v$. Indeed, let $\set{u,v} \in T'$, and assume without loss of generality that $u$ is the parent of $v$ in $T$. Then $P_u$ edge-intersects $P_v$ because they both use the edge connecting $u$ to its parent. Moreover they do not split, because their union is the path from $v$ to its ancestor at distance $3$. Therefore $P_u \sim P_v$. Conversely, assume that $P_u \sim P_v$. Then $P_u$ and $P_v$ edge-intersect. As every vertex is a starting vertex of at most one path and the paths are of length $2$, the second edge of one of the paths, say $P_v$ is the first edge of $P_u$, therefore $u$ is the parent of $v$ in $T$, i.e. $\set{u,v} \in T'$.
\end{proof}

Let $T$ be a tree with $k$ leaves and $\pi=(\pi_0,\ldots,\pi_{k-1})$ a cyclic permutation of the leaves. The \emph{tour} $(T,\pi)$ is the following set of $2k$ paths: $(T,\pi)$ contains $k$ \emph{long} paths, each of which connecting two consecutive leaves $\pi_i, \pi_{i+1 \mod k}$. $(T,\pi)$ contains $k$ \emph{short} paths, each of which connecting a leaf $\pi_i$ and its unique neighbor in $T$ (see Figure \ref{fig:eptn-holes}-c). Note that $\enpt((T,\pi))$ is a cycle.

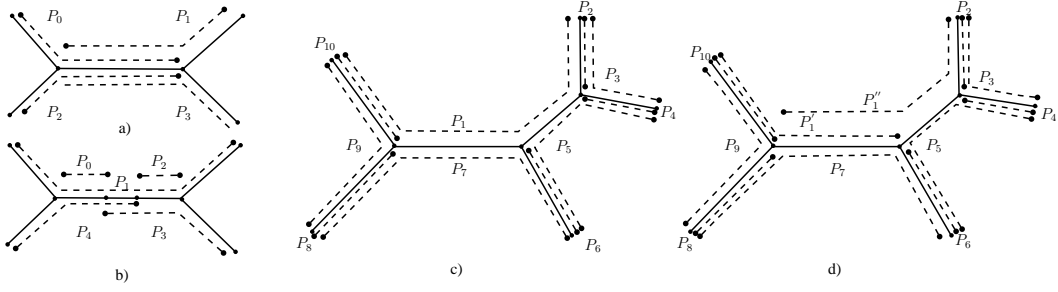
\begin{figure}[htbp]
\centering
\scalebox{0.5} 
{
\begin{pspicture}(0,-3.7676563)(28.261875,3.7676563)
\psline[linewidth=0.04cm](8.68,2.2092187)(10.36,-0.09078125)
\psline[linewidth=0.04cm](10.36,-0.09078125)(8.12,-2.3907812)
\psline[linewidth=0.04cm](10.34,-0.11078125)(13.72,-0.11078125)
\psline[linewidth=0.04cm](13.72,-0.11078125)(15.3,1.2492187)
\psline[linewidth=0.04cm](15.3,1.2492187)(15.28,3.3292189)
\psline[linewidth=0.04cm](15.3,1.2692188)(17.36,0.88921875)
\psline[linewidth=0.04cm](13.72,-0.11078125)(15.08,-2.5107813)
\psline[linewidth=0.04,linestyle=dashed,dash=0.16cm 0.16cm,dotsize=0.07055555cm 2.0]{**-**}(15.62,3.3692188)(15.64,1.4492188)(16.98,1.2092187)(17.46,1.1092187)
\psline[linewidth=0.04,linestyle=dashed,dash=0.16cm 0.16cm,dotsize=0.07055555cm 2.0]{**-**}(14.98,-2.6107812)(13.68,-0.41078126)(10.52,-0.41078126)(8.84,-2.1307812)(8.4,-2.5707812)
\psline[linewidth=0.04,linestyle=dashed,dash=0.16cm 0.16cm,dotsize=0.07055555cm 2.0]{**-**}(8.02,-2.2307813)(10.04,-0.09078125)(8.92,1.5292188)(8.52,2.1092188)
\psline[linewidth=0.04cm,fillcolor=black,linestyle=dashed,dash=0.16cm 0.16cm,dotsize=0.07055555cm 2.0]{**-**}(10.46,0.04921875)(8.78,2.3492188)
\psline[linewidth=0.04cm,fillcolor=black,linestyle=dashed,dash=0.16cm 0.16cm,dotsize=0.07055555cm 2.0]{**-**}(15.42,1.4092188)(15.4,3.3692188)
\psline[linewidth=0.04cm,linestyle=dashed,dash=0.16cm 0.16cm,dotsize=0.07055555cm 2.0]{**-**}(15.34,1.1692188)(17.34,0.7892187)
\psline[linewidth=0.04cm,fillcolor=black,linestyle=dashed,dash=0.16cm 0.16cm,dotsize=0.07055555cm 2.0]{**-**}(13.88,-0.15078124)(15.24,-2.4507813)
\psline[linewidth=0.04cm,linestyle=dashed,dash=0.16cm 0.16cm,dotsize=0.07055555cm 2.0]{**-**}(10.36,-0.25078124)(8.16,-2.5507812)
\usefont{T1}{ptm}{m}{n}
\rput(15.431406,3.5792189){$P_2$}
\usefont{T1}{ptm}{m}{n}
\rput(16.151405,1.7992188){$P_3$}
\usefont{T1}{ptm}{m}{n}
\rput(17.631407,0.7992188){$P_4$}
\usefont{T1}{ptm}{m}{n}
\rput(14.851406,-0.10078125){$P_5$}
\usefont{T1}{ptm}{m}{n}
\rput(15.631406,-2.7407813){$P_6$}
\usefont{T1}{ptm}{m}{n}
\rput(12.071406,-0.72078127){$P_7$}
\usefont{T1}{ptm}{m}{n}
\rput(7.9514065,-2.7007813){$P_8$}
\usefont{T1}{ptm}{m}{n}
\rput(9.271406,-0.08078125){$P_9$}
\usefont{T1}{ptm}{m}{n}
\rput(8.501407,2.6392188){$P_{10}$}
\psline[linewidth=0.04,linestyle=dashed,dash=0.16cm 0.16cm,dotsize=0.07055555cm 2.0]{**-**}(17.32,0.58921874)(15.3,1.0292188)(14.08,-0.13078125)(15.36,-2.3507812)
\usefont{T1}{ptm}{m}{n}
\rput(12.051406,0.55921876){$P_1$}
\psline[linewidth=0.04cm](18.72,2.1892188)(20.42,-0.09078125)
\psline[linewidth=0.04cm](20.42,-0.09078125)(18.22,-2.3707812)
\psline[linewidth=0.04cm](20.4,-0.11078125)(23.78,-0.11078125)
\psline[linewidth=0.04cm](23.78,-0.11078125)(25.36,1.2492187)
\psline[linewidth=0.04cm](25.36,1.2492187)(25.32,3.3492188)
\psline[linewidth=0.04cm](25.36,1.2692188)(27.44,0.94921875)
\psline[linewidth=0.04cm](23.78,-0.11078125)(25.16,-2.4907813)
\psline[linewidth=0.04,linestyle=dashed,dash=0.16cm 0.16cm,dotsize=0.07055555cm 2.0]{**-**}(25.62,3.3892188)(25.62,1.4892187)(26.98,1.2492187)(27.52,1.1892188)
\psline[linewidth=0.04,linestyle=dashed,dash=0.16cm 0.16cm,dotsize=0.07055555cm 2.0]{**-**}(24.88,-2.5707812)(23.656076,-0.33078125)(20.664951,-0.3700473)(18.913568,-2.0388546)(18.42,-2.5507812)
\psline[linewidth=0.04,linestyle=dashed,dash=0.16cm 0.16cm,dotsize=0.07055555cm 2.0]{**-**}(18.04,-2.2107813)(20.12,-0.05078125)(19.04,1.4092188)(18.56,1.9892187)
\psline[linewidth=0.04cm,fillcolor=black,linestyle=dashed,dash=0.16cm 0.16cm,dotsize=0.07055555cm 2.0]{**-**}(20.5,0.02921875)(18.82,2.3092186)
\psline[linewidth=0.04cm,fillcolor=black,linestyle=dashed,dash=0.16cm 0.16cm,dotsize=0.07055555cm 2.0]{**-**}(25.5,1.4092188)(25.44,3.3892188)
\psline[linewidth=0.04cm,fillcolor=black,linestyle=dashed,dash=0.16cm 0.16cm,dotsize=0.07055555cm 2.0]{**-**}(25.44,1.1292187)(27.4,0.7892187)
\psline[linewidth=0.04cm,fillcolor=black,linestyle=dashed,dash=0.16cm 0.16cm,dotsize=0.07055555cm 2.0]{**-**}(23.98,-0.19078125)(25.32,-2.4307814)
\psline[linewidth=0.04cm,fillcolor=black,linestyle=dashed,dash=0.16cm 0.16cm,dotsize=0.07055555cm 2.0]{**-**}(20.46,-0.33078125)(18.3,-2.4707813)
\usefont{T1}{ptm}{m}{n}
\rput(25.491405,3.5792189){$P_2$}
\usefont{T1}{ptm}{m}{n}
\rput(26.091406,1.7192187){$P_3$}
\usefont{T1}{ptm}{m}{n}
\rput(27.751406,0.7392188){$P_4$}
\usefont{T1}{ptm}{m}{n}
\rput(24.691406,-0.08078125){$P_5$}
\usefont{T1}{ptm}{m}{n}
\rput(25.471407,-2.6207812){$P_6$}
\usefont{T1}{ptm}{m}{n}
\rput(22.131407,-0.72078127){$P_7$}
\usefont{T1}{ptm}{m}{n}
\rput(18.071405,-2.7007813){$P_8$}
\usefont{T1}{ptm}{m}{n}
\rput(19.331406,-0.08078125){$P_9$}
\usefont{T1}{ptm}{m}{n}
\rput(18.481407,2.3792188){$P_{10}$}
\psline[linewidth=0.04,linestyle=dashed,dash=0.16cm 0.16cm,dotsize=0.07055555cm 2.0]{**-**}(27.36,0.6092188)(25.28,1.0092187)(24.06,-0.01078125)(25.48,-2.3707812)
\psline[linewidth=0.04,linestyle=dashed,dash=0.16cm 0.16cm,dotsize=0.07055555cm 2.0]{**-**}(18.98,2.3892188)(20.52,0.18921874)(22.24,0.16921875)(23.8,0.16921875)
\psline[linewidth=0.04,linestyle=dashed,dash=0.16cm 0.16cm,dotsize=0.07055555cm 2.0]{**-**}(20.62,0.80921876)(23.86,0.82921875)(25.08,1.8292187)(25.08,3.3492188)
\usefont{T1}{ptm}{m}{n}
\rput(21.331406,0.5392187){$P_1^{'}$}
\usefont{T1}{ptm}{m}{n}
\rput(23.051407,1.1192187){$P_1^{''}$}
\psline[linewidth=0.04,linestyle=dashed,dash=0.16cm 0.16cm,dotsize=0.07055555cm 2.0]{**-**}(9.0,2.3892188)(10.52,0.30921876)(13.64,0.28921875)(14.98,1.3692187)(14.96,3.3692188)
\usefont{T1}{ptm}{m}{n}
\rput(12.006406,-3.4007812){c)}
\usefont{T1}{ptm}{m}{n}
\rput(22.056562,-3.4007812){d)}
\psdots[dotsize=0.12](8.68,2.1892188)
\psdots[dotsize=0.12](10.34,-0.09078125)
\psdots[dotsize=0.12](8.16,-2.3707812)
\psdots[dotsize=0.12](15.06,-2.4707813)
\psdots[dotsize=0.12](13.72,-0.11078125)
\psdots[dotsize=0.12](15.3,1.2692188)
\psdots[dotsize=0.12](17.32,0.9092187)
\psdots[dotsize=0.12](15.28,3.3092186)
\psdots[dotsize=0.12](18.76,2.1492188)
\psdots[dotsize=0.12](20.42,-0.09078125)
\psdots[dotsize=0.12](18.22,-2.3707812)
\psdots[dotsize=0.12](23.78,-0.11078125)
\psdots[dotsize=0.12](25.16,-2.4707813)
\psdots[dotsize=0.12](27.38,0.96921873)
\psdots[dotsize=0.12](25.32,3.3292189)
\psline[linewidth=0.04cm](6.26,-0.07078125)(4.66,-1.4907813)
\psline[linewidth=0.04cm](4.66,-1.4907813)(6.12,-2.9107811)
\psline[linewidth=0.04cm](4.66,-1.4907813)(1.32,-1.4707812)
\psline[linewidth=0.04cm](0.1,-0.05078125)(1.34,-1.4707812)
\psline[linewidth=0.04cm](1.34,-1.4907813)(0.06,-2.7107813)
\psline[linewidth=0.04,linestyle=dashed,dash=0.16cm 0.16cm,dotsize=0.07055555cm 2.0]{**-**}(0.28,-0.01078125)(1.46,-1.2707813)(4.66,-1.2707813)(6.12,0.04921875)
\psline[linewidth=0.04,linestyle=dashed,dash=0.16cm 0.16cm,dotsize=0.07055555cm 2.0]{**-**}(0.22,-2.8707812)(1.6312675,-1.6307813)(3.56,-1.6307813)
\psline[linewidth=0.04,linestyle=dashed,dash=0.16cm 0.16cm,dotsize=0.07055555cm 2.0]{**-**}(2.58,-1.8907813)(4.5984616,-1.8707813)(5.86,-2.9907813)
\psline[linewidth=0.04cm,fillcolor=black,linestyle=dashed,dash=0.16cm 0.16cm,dotsize=0.07055555cm 2.0]{**-**}(1.48,-0.85078126)(2.8,-0.85078126)
\psline[linewidth=0.04cm,fillcolor=black,linestyle=dashed,dash=0.16cm 0.16cm,dotsize=0.07055555cm 2.0]{**-**}(3.5,-0.8907812)(4.72,-0.87078124)
\psdots[dotsize=0.12](0.12,-0.07078125)
\psdots[dotsize=0.12](0.06,-2.7107813)
\psdots[dotsize=0.12](6.12,-2.9107811)
\psdots[dotsize=0.12](6.26,-0.07078125)
\psdots[dotsize=0.12](2.68,-1.4707812)
\usefont{T1}{ptm}{m}{n}
\rput(2.0914063,-0.5207813){$P_0$}
\usefont{T1}{ptm}{m}{n}
\rput(3.1114063,-1.0807812){$P_1$}
\usefont{T1}{ptm}{m}{n}
\rput(4.0914063,-0.5207813){$P_2$}
\psdots[dotsize=0.12](3.5,-1.4707812)
\usefont{T1}{ptm}{m}{n}
\rput(4.0914063,-2.4007812){$P_3$}
\usefont{T1}{ptm}{m}{n}
\rput(2.0914063,-2.4007812){$P_4$}
\usefont{T1}{ptm}{m}{n}
\rput(3.1029687,-3.5407813){b)}
\psline[linewidth=0.04cm](6.32,3.3692188)(4.72,1.9492188)
\psline[linewidth=0.04cm](4.72,1.9492188)(6.18,0.52921873)
\psline[linewidth=0.04cm](4.72,1.9492188)(1.38,1.9692187)
\psline[linewidth=0.04cm](0.16,3.3892188)(1.4,1.9692187)
\psline[linewidth=0.04cm](1.4,1.9492188)(0.12,0.7292187)
\psdots[dotsize=0.12](0.18,3.3692188)
\psdots[dotsize=0.12](0.12,0.7292187)
\psdots[dotsize=0.12](6.18,0.52921873)
\psdots[dotsize=0.12](6.32,3.3692188)
\usefont{T1}{ptm}{m}{n}
\rput(1.3114063,3.3592188){$P_0$}
\usefont{T1}{ptm}{m}{n}
\rput(4.731406,3.3592188){$P_1$}
\usefont{T1}{ptm}{m}{n}
\rput(1.3114063,0.77921873){$P_2$}
\usefont{T1}{ptm}{m}{n}
\rput(4.731406,0.77921873){$P_3$}
\usefont{T1}{ptm}{m}{n}
\rput(3.1860938,0.35921875){a)}
\psdots[dotsize=0.12](1.32,-1.4907813)
\psdots[dotsize=0.12](4.7,-1.5107813)
\psline[linewidth=0.04,linestyle=dashed,dash=0.16cm 0.16cm,dotsize=0.07055555cm 2.0]{**-**}(0.36,3.4492188)(1.5,2.1892188)(4.64,2.1892188)
\psline[linewidth=0.04,linestyle=dashed,dash=0.16cm 0.16cm,dotsize=0.07055555cm 2.0]{**-**}(1.54,2.5692186)(4.72,2.5692186)(5.88,3.5892189)
\psline[linewidth=0.04,linestyle=dashed,dash=0.16cm 0.16cm,dotsize=0.07055555cm 2.0]{**-**}(0.46,0.76921874)(1.48,1.7492187)(4.68,1.7692188)
\psline[linewidth=0.04,linestyle=dashed,dash=0.16cm 0.16cm](1.58,1.5292188)(4.74,1.5492188)(5.96,0.38921875)
\psdots[dotsize=0.12](1.4,1.9692187)
\psdots[dotsize=0.12](4.72,1.9492188)
\psdots[dotsize=0.12](25.38,1.2492187)
\end{pspicture} 
}
\caption{a) A minimal representation of $C_4$ b) A minimal representation of $C_5$ c) A tour representation of the even hole $C_{10}$, d) A representation of the odd hole $C_{11}$.}
\label{fig:eptn-holes}
\end{figure}

A \emph{planar embedding} of a tour is a planar embedding of the underlying tree such that any two paths of the tour do not cross each other. A tour is \emph{planar} if there exists a planar embedding of it. The tour in Figure \ref{fig:eptn-holes}-c is a planar embedding of a tour. Note that a tour $(T,\pi)$ is planar if and only if $\pi$ corresponds to the order in which the leaves are encountered by some DFS traversal of $T$.

\begin{lemma}\label{lem:eptn-ring}
Every cycle $C_k$ is an $\enpt$ graph.
\end{lemma}
\begin{proof}
$C_3=K_3$ is an $\enpt$ graph by Lemma \ref{lem:eptn-clique}. As for $C_4$ and $C_5$, possible $\enpt$ representations are shown in Figure \ref{fig:eptn-holes}-(a,b), respectively. Any even hole $C_{2k}$, $(k \geq 3)$ is an $\enpt$ graph. Indeed, for any tree $T$ with $k$ leaves, and a cyclic permutation $\pi$ of its leaves, the tour $(T,\pi)$ constitutes an $\enpt$ representation of $C_{2k}$. Any odd hole $C_{2k+1}$, $(k \geq 3)$ is an $\enpt$ graph. Let $T$ be a tree with $k$ leaves. Split any long path of some tour $(T,\pi)$ into two edge-intersecting sub-paths such that no chord is created (if necessary subdivide an edge of the tree into two edges) (see Figure \ref{fig:eptn-holes}-d). The set of $2k+1$ paths obtained in this way constitutes an $\enpt$ representation for $C_{2k+1}$.
\end{proof}

\begin{lemma}\label{lem:eptn-ept}
$\enpt \nsubseteq \ept$.
\end{lemma}
\begin{proof}
Consider the graph $W_{5,1}$ consisting of a vertex adjacent to every other vertex of a $C_5$. In \cite{GJ85} it is shown that a vertex of an $\ept$ graph has at most $4$ neighbours in a cycle. Therefore, $W_{5,1} \notin \ept$. On the other hand Figure \ref{fig:eptn-w51} depicts an $\enpt$ representation of $W_{5,1}$.
\end{proof}

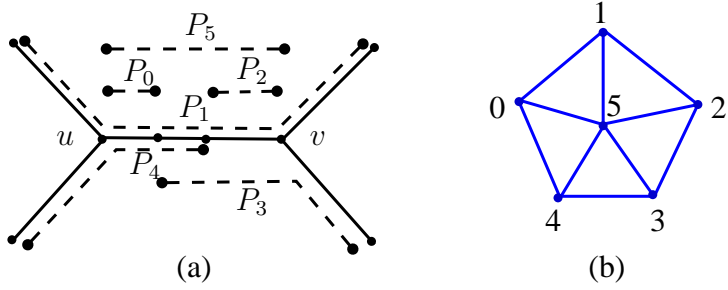
\begin{figure}[htbp]
\centering
\scalebox{1} 
{
\begin{pspicture}(0,-1.9101562)(9.585,1.9101562)
\definecolor{color137}{rgb}{0.0,0.0,0.8}
\psline[linewidth=0.04cm](0.08,1.3267188)(1.28,0.04671875)
\psline[linewidth=0.04cm](1.28,0.04671875)(0.06,-1.3132813)
\psline[linewidth=0.04cm](1.26,0.06671875)(3.66,0.04671875)
\psline[linewidth=0.04cm](3.66,0.04671875)(4.9,1.2867187)
\psline[linewidth=0.04cm](3.64,0.04671875)(4.88,-1.3532813)
\psline[linewidth=0.04,linestyle=dashed,dash=0.16cm 0.16cm,dotsize=0.07055555cm 2.0]{**-**}(0.22,-1.4132812)(1.4462337,-0.09328125)(2.68,-0.09328125)
\psline[linewidth=0.04,linestyle=dashed,dash=0.16cm 0.16cm,dotsize=0.07055555cm 2.0]{**-**}(0.18,1.4067187)(1.28,0.20671874)(3.6,0.18671875)(4.78,1.3667188)
\psline[linewidth=0.04,linestyle=dashed,dash=0.16cm 0.16cm,dotsize=0.07055555cm 2.0]{**-**}(1.98,-0.53328127)(3.84,-0.5132812)(4.64,-1.4732813)
\psline[linewidth=0.04cm,fillcolor=black,linestyle=dashed,dash=0.16cm 0.16cm,dotsize=0.07055555cm 2.0]{**-**}(1.26,0.68671876)(2.04,0.68671876)
\psline[linewidth=0.04cm,fillcolor=black,linestyle=dashed,dash=0.16cm 0.16cm,dotsize=0.07055555cm 2.0]{**-**}(2.66,0.6667187)(3.66,0.68671876)
\usefont{T1}{ptm}{m}{n}
\rput(2.4914062,0.45671874){$P_1$}
\usefont{T1}{ptm}{m}{n}
\rput(1.7514062,0.93671876){$P_0$}
\usefont{T1}{ptm}{m}{n}
\rput(3.2514062,0.95671874){$P_2$}
\usefont{T1}{ptm}{m}{n}
\rput(1.8314062,-0.30328125){$P_4$}
\usefont{T1}{ptm}{m}{n}
\rput(3.2514062,-0.78328127){$P_3$}
\usefont{T1}{ptm}{m}{n}
\rput(7.886875,1.7367188){1}
\usefont{T1}{ptm}{m}{n}
\rput(7.2609377,-1.0389062){4}
\usefont{T1}{ptm}{m}{n}
\rput(8.647656,-1.0432812){3}
\usefont{T1}{ptm}{m}{n}
\rput(6.517031,0.47234374){0}
\usefont{T1}{ptm}{m}{n}
\rput(9.458593,0.47671875){2}
\psdots[dotsize=0.12](0.08,1.3267188)
\psdots[dotsize=0.12](0.06,-1.2932812)
\psdots[dotsize=0.12](1.26,0.04671875)
\psdots[dotsize=0.12](3.64,0.04671875)
\psdots[dotsize=0.12](4.88,1.2667187)
\psdots[dotsize=0.12](4.84,-1.3132813)
\psdots[dotsize=0.12](2.0,0.06671875)
\psdots[dotsize=0.12](2.64,0.04671875)
\usefont{T1}{ptm}{m}{n}
\rput(0.77140623,0.05671875){$u$}
\usefont{T1}{ptm}{m}{n}
\rput(4.1314063,0.05671875){$v$}
\psline[linewidth=0.04,linecolor=blue](7.92,1.4867188)(6.8,0.5467188)(7.34,-0.7132813)(8.6,-0.7132813)(9.18,0.50671875)(7.94,1.4867188)
\usefont{T1}{ptm}{m}{n}
\rput(2.4851563,-1.6832813){(a)}
\usefont{T1}{ptm}{m}{n}
\rput(7.9751563,-1.6832813){(b)}
\psline[linewidth=0.04cm,linecolor=blue](7.92,1.4867188)(7.92,0.24671875)
\psline[linewidth=0.04cm,linecolor=blue](7.9,0.24671875)(6.8,0.56671876)
\psline[linewidth=0.04cm,linecolor=blue](7.94,0.24671875)(7.36,-0.7132813)
\psline[linewidth=0.04cm,linecolor=blue](7.92,0.24671875)(8.58,-0.69328123)
\psline[linewidth=0.04cm,linecolor=blue](7.92,0.24671875)(9.18,0.50671875)
\psdots[dotsize=0.12,linecolor=color137](7.92,0.22671875)
\psdots[dotsize=0.12,linecolor=color137](6.8,0.5467188)
\psdots[dotsize=0.12,linecolor=color137](7.32,-0.73328125)
\psdots[dotsize=0.12,linecolor=color137](8.58,-0.69328123)
\psdots[dotsize=0.12,linecolor=color137](9.18,0.50671875)
\psdots[dotsize=0.12,linecolor=color137](7.92,1.4667188)
\usefont{T1}{ptm}{m}{n}
\rput(8.069531,0.50109375){5}
\psline[linewidth=0.04cm,fillcolor=black,linestyle=dashed,dash=0.16cm 0.16cm,dotsize=0.07055555cm 2.0]{**-**}(1.24,1.2467188)(3.76,1.2467188)
\usefont{T1}{ptm}{m}{n}
\rput(2.5714064,1.5367187){$P_5$}
\end{pspicture} 
}
\caption {An $\enpt$ representation of $W_{5,1}$.}
\label{fig:eptn-w51}
\end{figure}

In \cite{BESZ13-TR2} we present a family of non-$\enpt$ graphs. However, these graphs are non-$\ept$ graphs. Therefore, the question whether $\ept \nsubseteq \enpt$ holds is open.

\section{Representations of $\ept$, $\enpt$ Pairs: Basic Properties}\label{sec:EPTNBasics}
In this section we develop the basic tools that we use in subsequent sections towards our goal of characterizing representations of $\enpt, \ept$ pairs. We define an equivalence relation on representations, namely two representations will be equivalent in this relation if they are representations of the same pair. We also define a partial order on representations. In this work, we focus on finding representations that are minimal with respect to this partial order. We define the contraction operation on pairs, and the union operation on representations. The contraction operation is a restricted variant of graph contraction operation that operates on both graphs of a pair. The union operation is the operation of replacing two paths by their union whenever possible.

\runningtitle{Equivalent and minimal representations:}
We say that the representations $\repn{1}$ and $\repn{2}$ are \emph{equivalent}, and denote it by $\repn{1} \approxeq \repn{2}$, if their corresponding $\ept$ and $\enpt$ graphs are isomorphic under the same isomorphism (in other words, if they constitute representations of the same pair of graphs $(G,G')$).

We write $\repn{1} \rightsquigarrow \repn{2}$ if $\repn{2}$ can be obtained from $\repn{1}$ by one of the following two operations that we term \emph{minifying operations}:

\begin{itemize}
\item{} Contraction of an edge $e$ of $T_1$ (and of all the paths in $\pp_1$ using $e$)
\item{} Removal of an initial edge (\emph{tail}) of a path in $\pp_1$.
\end{itemize}

The partial order $\gtrsim$ is the reflexive-transitive closure of the relation $\rightsquigarrow$, and $\repn{1} \lesssim \repn{2}$ is equivalent to $\repn{2} \gtrsim \repn{1}$.
$\rep$ is a \emph{minimal} representation if it is minimal in the partial order $\lesssim$ restricted to its equivalence class $[\rep]_\approxeq$ i.e., over all the representations representing the same pair as $\rep$. Throughout the work we aim at characterizing minimal representations. 

\runningtitle{$\ept$ Holes:}
\begin{lemma}\label{lem:NoBlueHole}
A hole of size at least $4$ of an $\ept$ graph does not contain blue (i.e. $\enpt$) edges.
\end{lemma}
\begin{proof}
Consider the pie representation of some hole of size at least $4$ of an $\ept$ graph(Theorem \ref{thm:golumbicpie}). For any two paths $P_p,P_q$ of this pie, we have either $P_p \nsim P_q$ or $P_p \parallel P_q$, therefore $\set{p,q}$ is not an $\enpt$ edge.
\end{proof}
Combining with Theorem \ref{thm:golumbicpie}, we obtain the following characterization of pairs $(C_k, G')$:
\begin{itemize}
\item{$k>3$.} In this case $C_k$ is represented by a pie. Therefore $G'$ is an independent set. In other words $C_k$ consists of red edges. We term such a hole, a red hole.
\item{$k=3$ and $C_k$ consists of red edges ($G'$ is an independent set).} We term such a hole a red triangle.
\item{$k=3$ and $C_k$ contains exactly one $\enpt$ (blue) edge ($G' = P_1 \cup P_2$).} We term such a hole a $BRR$ triangle, and its representation is an edge-clique.
\item{$k=3$ and $C_k$ contains two $\enpt$ (blue) edges ($G' = P_3$).} We term such a hole a $BBR$ triangle, and its representation is an edge-clique.
\item{$k=3$ and $C_k$ consists of blue edges ($G' = C_3$).} We term such a hole a blue triangle.
\end{itemize}

\runningtitle{$\ept$ contraction:}
Let $\rep$ be a representation and $P_p,P_q \in \pp$ such that $P_p \sim P_q$. We denote by $\contractppq$ the representation that is obtained from $\rep$ by replacing the two paths $P_p,P_q$ by the path $P_p \cup P_q$, i.e. $\contractppq \defined \left<T,\pp \setminus \set{P_p,P_q} \cup \set{P_p \cup P_q}\right>$. We term this operation a \emph{union}, and note the following important property of split vertices with respect to the union operation:
\begin{observation}\label{obs:splitofunion}
For every $P_p,P_q,P_r \in \pp$ such that $P_p \sim P_q$, $\split(P_p \cup P_q, P_r) = \split(P_p,P_r) \cup \split(P_q,P_r)$.
\end{observation}

\begin{lemma}\label{lem:contraction-union}
Let $\rep$ be a representation for the pair $(G,G')$, and let $e=\set{p,q} \in E(G')$. Then
$\contractge$ is an $\ept$ graph. Moreover $\contractge=\eptg{\contractppq}$.
\end{lemma}
\begin{proof}
Let $s$ be the vertex of $\contractge$ created by the contraction of $e$. We claim that $s$ corresponds to the path $P_s = P_p \cup P_q$. Consider a path $P_r \in \pp \setminus \set{P_p,P_q}$. We observe that $\set{r,s} \in E(\contractge) \iff \set{r,p} \in E(G)$ or $\set{r,q} \in E(G)$ (by definition of the contraction operation) $\iff P_r$ edge-intersects with at least one of $P_p$ and $P_q$ in $T$ (because $G=\eptg{\pp}$) $\iff$ $P_r$ edge-intersects $P_p \cup P_q = P_s$ in $T \iff \set{r, s} \in E(\eptg{\contractppq})$.
\end{proof}

We now extend the definition of the contraction operation to pairs. Based on Observation \ref{obs:splitofunion}, the contraction of an $\enpt$ edge does not necessarily preserve $\enpt$ edges. More concretely, let $P_p$,$P_q$ and $P_{q'}$ such that $P_p \sim P_q$, $P_p \sim P_{q'}$ and $P_q \nsim P_{q'}$. Then $\contract{G'}{p,q}$ is not isomorphic to $\enptg{\contractppq}$ as $ \left\{q',p.q\right\} \notin E(\enptg{\contractppq})$. Let $(G,G')$ be a pair and $e \in E(G')$. If for every edge $e' \in E(G')$ incident to $e$, the edge $e''= e \triangle e'$ (forming a triangle together with $e$ and $e'$) is not an edge of $E(G)\setminus E(G')$ (i.e. not a red edge) then $\contract{(G,G')}{e} \defined (\contractge,\contract{G'}{e})$, otherwise $\contract{(G,G')}{e}$ is undefined. Whenever $\contract{(G,G')}{e}$ is defined we say that $(G,G')$ is \emph{contractible} on $e$, and when there is no ambiguity about the pair under consideration we say that $e$ is \emph{contractible}. A pair $(G,G')$ is \emph{contractible} if it contains at least one contractible edge. Clearly, $(G,G')$ is non-contractible if and only if every edge of $G'$ is contained in at least one $BBR$ triangle.

Our goal in this work is to characterize the representations of $\enpt$ holes. More precisely we characterize representations of pairs $(G,C_n)$ where $C_n$ is a Hamiltonian cycle of $G$. For this purpose we define the following problem.

\begin{center} \fbox{\begin{minipage}{11.9cm}
\noindent  \prb

{\bf Input:} A pair $(G,C_n)$ where $C_n$ is a Hamiltonian cycle of $G$

{\bf Output:} A minimal representation $\rep$ of $(G,C_n)$ if such a representation exists, ``NO'' otherwise.

\end{minipage}}\\
\end{center}

The $\enpt$ representations of $C_3$ are characterized by Lemma \ref{lem:eptn-clique}. Therefore we assume $n > 3$, which implies that $(G,C_n)$ does not contain blue triangles. In the sequel we confine ourselves to pairs $(G,C_n)$ and representations $\rep$ satisfying the following three assumptions:
\begin{itemize}
\item{$(P1)$:} $(G,C_n)$ is not contractible.
\item{$(P2)$:} $(G,C_n)$ is $(K_4, P_4)$-free, i.e., it does not contain an induced sub-pair isomorphic to a $(K_4, P_4)$.
\item{$(P3)$:} Every red triangle of $(G,C_n)$ is a claw-clique, i.e. corresponds to a pie of $\rep$.
\end{itemize}

Note that $(P1)$ and $(P2)$ are assumptions about the pair $(G,C)$ and $(P3)$ is an assumption about the representation $\rep$. We say that $(P3)$ holds for a pair $(G,C)$ if it has a representation $\rep$ satisfying $(P3)$. It will be convenient to define the following problem.

\begin{center} \fbox{\begin{minipage}{11.9cm}
\noindent  \prbpthree

{\bf Input:} A pair $(G,C_n)$ where $C_n$ is a Hamiltonian cycle of $G$ and $n \geq 4$.

{\bf Output:} A minimal representation $\rep$ of $(G,C_n)$ that satisfies $(P3)$ if such a representation exists, ``NO '' otherwise.

\end{minipage}}\\
\end{center}

In this work we consider instances of $\prbpthree$ satisfying $(P1)$ and $(P2)$. Without loss of generality we let $V(G)=V(C_n)=\set{0,1,\ldots,n-1}$ where the numbering of the vertices is consistent with their order in $C$. All arithmetic operations on vertex numbers are done modulo $n$. We denote the corresponding set of paths in the representation as $\pp=\set{P_0, \ldots, P_{n-1}}$.

\section{Representation of $\enpt$ Holes}\label{sec:unionminimal}
In this section we characterize the minimal representations of $(G,C)$ pairs satisfying $(P1)$, $(P2)$ and $(P3)$. To achieve this goal we present an algorithm solving the $\prbpthree$ problem for instances satisfying $(P1)$ and $(P2)$.
In Section \ref{subsec:c4} we handle the case $n=4$. In Section \ref{subsec:wdt} we analyze properties of weak dual trees based on which, in Section \ref{subsec:noIntermediateVertices} we present an algorithm for the case $n>4$. $C_4$ is exceptional because all its representations satisfy assumptions $(P1-3)$, but some of our results that we prove for $n>4$ fail to hold in this case.

\subsection{The pair $(G,C_4)$}\label{subsec:c4}

\begin{lemma}\label{lem:C4}
(i) All the representations of $(G,C_4)$ satisfy assumptions $(P1-3)$, (ii) $G$ is one of the two graphs in Figure \ref{fig:eptn-c4}, and (iii) each of these two graphs has a unique minimal representation (also depicted in Figure \ref{fig:eptn-c4}).
\end{lemma}
\begin{proof}
(i) $(G,C_4)$ is clearly $(K_4,P_4)$-free. Moreover it satisfies $(P3)$ vacuously, because it does not contain any red triangle. $G \neq C_4$, because otherwise $C_4$ would constitute a blue hole of length $4$, contradicting Lemma \ref{lem:NoBlueHole}. Without loss of generality let $\set{1,3}$ be a red edge of $G$. We observe that $\set{1,3}$ is incident to all the edges of $C_4$, therefore $(G,C_4)$ is not contractible, so it satisfies $(P1)$.

(ii) Depending on whether or not $\set{0,2} \in E(G)$, $G$ is one of the two graphs in Figure \ref{fig:eptn-c4}.

(iii) Consider a representation $\rep$ of $(G,C_4)$, and consider the path $P=P_1 \cap P_3$. Let $e_0$ (resp. $e_2$) be an edge defining the edge-clique $\set{1,3,0}$ (resp. $\set{1,3,2}$). Both of $e_0$ and $e_2$ are in $P$. Let $u \in \split(P_1,P_3)$. $u$ is an endpoint of $P$. As $P_0$ edge-intersects $P$ (at $e_0$), it can not cross $u$, because in this case it has to split from at least one of $P_1,P_3$ at $u$. The same holds for $P_2$. Therefore neither one of $P_0,P_2$ crosses a vertex of $\split(P_1,P_3)$. We consider two cases:
(a) $G$ is isomorphic to $K_4$. Then there is one edge defining the clique, i.e. without loss of generality $e_0=e_2$. If $\abs{\split(P_1,P_3)}=2$ then, none of these two vertices can be crossed by $P_0$ or $P_2$. Therefore $P_0 \subseteq P$ and $P_2 \subseteq P$, we conclude that they can not split, a contradiction. Therefore $\split(P_1,P_3)$ consists of one vertex that is not crossed by $P_0$ and $P_2$. We conclude that $P_0$ and $P_2$ cross the other endpoint of $P$ and split. The representation in Figure \ref{fig:eptn-c4} (a) is the only minimal representation satisfying these conditions.
(b) $G$ is not isomorphic to $K_4$. Therefore $e_0 \neq e_2$, and without loss of generality $e_0 \in P_0 \setminus P_2, e_2 \in P_2 \setminus P_0$. $P$ has at least one endpoint $u$ in $\split(P_1,P_3)$. Without loss of generality $e_0$ is closer to $u$ than $e_2$. Therefore $P_0$ lies between $u$ and $e_2$, and $P_2$ starts after $P_1$ and crosses $e_2$. The representation in Figure \ref{fig:eptn-c4} (b) is the only minimal representation satisfying these conditions.
\end{proof}

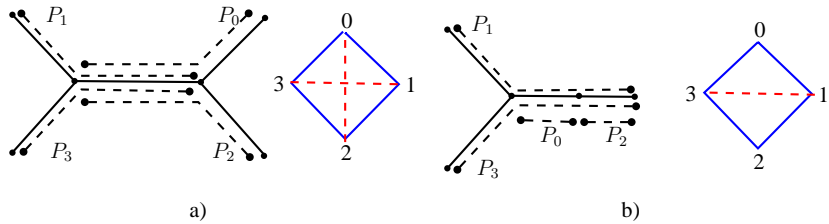
\begin{figure}[htbp]
\centering
\scalebox{0.7} 
{
\begin{pspicture}(0,-2.0676563)(15.595938,2.0676563)
\psline[linewidth=0.04cm](0.06,1.9092188)(1.26,0.62921876)
\psline[linewidth=0.04cm](1.26,0.62921876)(0.04,-0.73078126)
\psline[linewidth=0.04cm](1.24,0.64921874)(3.64,0.62921876)
\psline[linewidth=0.04cm](3.64,0.62921876)(4.88,1.8692187)
\psline[linewidth=0.04cm](3.62,0.62921876)(4.86,-0.7707813)
\psline[linewidth=0.04,linestyle=dashed,dash=0.16cm 0.16cm,dotsize=0.07055555cm 2.0]{**-**}(0.18,1.9892187)(1.3,0.74921876)(3.58,0.74921876)
\psline[linewidth=0.04,linestyle=dashed,dash=0.16cm 0.16cm,dotsize=0.07055555cm 2.0]{**-**}(0.22,-0.75078124)(1.3701298,0.48921874)(3.5,0.44921875)
\psline[linewidth=0.04,linestyle=dashed,dash=0.16cm 0.16cm,dotsize=0.07055555cm 2.0]{**-**}(1.36,0.96921873)(3.5780392,0.96921873)(4.6,1.9892187)
\psline[linewidth=0.04,linestyle=dashed,dash=0.16cm 0.16cm,dotsize=0.07055555cm 2.0]{**-**}(1.36,0.24921875)(3.5904,0.24921875)(4.62,-0.85078126)
\psline[linewidth=0.04,linecolor=blue](6.34,1.5892187)(5.36,0.6092188)(6.38,-0.45078126)(7.4,0.58921874)(6.38,1.5692188)(6.4,1.5492188)
\psline[linewidth=0.04cm,linecolor=red,linestyle=dashed,dash=0.16cm 0.16cm](6.38,1.4292188)(6.38,-0.5107812)
\psline[linewidth=0.04cm,linecolor=red,linestyle=dashed,dash=0.16cm 0.16cm](5.34,0.62921876)(7.38,0.58921874)
\usefont{T1}{ptm}{m}{n}
\rput(4.1714063,1.8792187){$P_0$}
\usefont{T1}{ptm}{m}{n}
\rput(0.9114063,1.8792187){$P_1$}
\usefont{T1}{ptm}{m}{n}
\rput(4.0714064,-0.7007812){$P_2$}
\usefont{T1}{ptm}{m}{n}
\rput(0.99140626,-0.7007812){$P_3$}
\usefont{T1}{ptm}{m}{n}
\rput(6.4170313,1.7992188){0}
\usefont{T1}{ptm}{m}{n}
\rput(7.626875,0.57921875){1}
\usefont{T1}{ptm}{m}{n}
\rput(6.398594,-0.7007812){2}
\usefont{T1}{ptm}{m}{n}
\rput(5.1276565,0.5992187){3}
\psline[linewidth=0.04cm](8.34,1.6292187)(9.54,0.34921876)
\psline[linewidth=0.04cm](9.54,0.34921876)(8.32,-1.0107813)
\psline[linewidth=0.04cm](9.52,0.36921874)(11.92,0.34921876)
\psline[linewidth=0.04,linestyle=dashed,dash=0.16cm 0.16cm,dotsize=0.07055555cm 2.0]{**-**}(8.46,1.7092187)(9.58,0.46921876)(11.88,0.48921874)
\psline[linewidth=0.04,linestyle=dashed,dash=0.16cm 0.16cm,dotsize=0.07055555cm 2.0]{**-**}(8.46,-1.0907812)(9.640235,0.18921874)(11.98,0.16921875)
\psline[linewidth=0.04,linecolor=blue](14.22,1.3892188)(13.2,0.42921874)(14.22,-0.63078123)(15.24,0.40921876)(14.22,1.3892188)(14.24,1.3692187)
\psline[linewidth=0.04cm,linecolor=red,linestyle=dashed,dash=0.16cm 0.16cm](13.3,0.42921874)(15.3,0.38921875)
\usefont{T1}{ptm}{m}{n}
\rput(10.291407,-0.40078124){$P_0$}
\usefont{T1}{ptm}{m}{n}
\rput(8.991406,1.7192187){$P_1$}
\usefont{T1}{ptm}{m}{n}
\rput(11.531406,-0.40078124){$P_2$}
\usefont{T1}{ptm}{m}{n}
\rput(9.091406,-1.0207813){$P_3$}
\usefont{T1}{ptm}{m}{n}
\rput(14.257031,1.6192187){0}
\usefont{T1}{ptm}{m}{n}
\rput(15.466875,0.39921874){1}
\usefont{T1}{ptm}{m}{n}
\rput(14.238594,-0.88078123){2}
\usefont{T1}{ptm}{m}{n}
\rput(12.947657,0.41921875){3}
\psline[linewidth=0.04cm,linestyle=dashed,dash=0.16cm 0.16cm,dotsize=0.07055555cm 2.0]{**-**}(9.64,-0.09078125)(10.78,-0.13078125)
\psline[linewidth=0.04cm,linestyle=dashed,dash=0.16cm 0.16cm,dotsize=0.07055555cm 2.0]{*-*}(10.92,-0.13078125)(11.82,-0.13078125)
\psdots[dotsize=0.12](0.06,1.9092188)
\psdots[dotsize=0.12](0.06,-0.7107813)
\psdots[dotsize=0.12](4.84,-0.7707813)
\psdots[dotsize=0.12](4.86,1.8492187)
\psdots[dotsize=0.12](3.64,0.62921876)
\psdots[dotsize=0.12](1.24,0.64921874)
\psdots[dotsize=0.12](8.32,-1.0107813)
\psdots[dotsize=0.12](8.32,1.6292187)
\psdots[dotsize=0.12](9.54,0.36921874)
\psdots[dotsize=0.12](11.88,0.34921876)
\psdots[dotsize=0.12](10.82,0.36921874)
\usefont{T1}{ptm}{m}{n}
\rput(3.5860937,-1.8407812){a)}
\usefont{T1}{ptm}{m}{n}
\rput(11.802969,-1.8407812){b)}
\end{pspicture} 
}
\caption{The two minimal $\enpt$ representations of $C_4$.}
\label{fig:eptn-c4}
\end{figure}

\subsection{Weak Dual Trees}\label{subsec:wdt}
\newcommand{\algwdt}{\textsc{FindWeakDualTree}}

We extend the definition of the weak dual tree of Hamiltonian outerplanar graphs to any Hamiltonian graph as follows. Given a pair $(G,C)$ where $C$ is a Hamiltonian cycle of $G$, a weak dual tree of $(G,C)$ is the weak dual tree $\wdtg$ of an arbitrary Hamiltonian maximal outerplanar subgraph $\opg$ of $G$. $\opg$ can be built by starting from $C$ and adding to it arbitrarily chosen chords from $G$ as long as such chords exist and the resulting graph is planar. We note that under the assumptions of this paper, i.e. when $(P1-3)$ hold, $G$ will be shown to be outerplanar, and therefore there is actually one weak dual tree. 

By definition of a dual graph, vertices of $\wdtg$ correspond to faces of $\opg$. By maximality, the faces of $\opg$ correspond to holes of $G$. The degree of a vertex of $\wdtg$ is the number of red edges in the corresponding face of $\opg$. To emphasize the difference, for an outerplanar graph $G$ we will refer to \emph{the} weak dual tree of $G$, whereas for a (not necessarily outerplanar) graph $G$ we will refer to \emph{a} weak dual tree of $G$.

We proceed with observations on $\wdtg$:
\begin{itemize}
\item Edges of $\wdtg$ correspond to red edges of $\opg$ (by definition of a weak dual graph, and observing that the edges of the unbounded face are exactly the blue edges).
\item The degree of a vertex of $\wdtg$ is the number of red edges in the corresponding face of $\opg$, therefore the leaves (resp. intermediate vertices, junctions) of $\wdtg$ correspond to $BBR$ triangles (resp. $BRR$ triangles, red holes) of $(G,C)$ (recalling Lemma \ref{lem:NoBlueHole}).
\item $\abs{V(G)}=\abs{V(C)}=\abs{E(C)}= 2 \ell + i$ where $\ell$ is the number of leaves of $\wdtg$ and $i$ is the number of its intermediate vertices.
\end{itemize}

\begin{lemma}\label{lem:BBRTrianglesDontOverlap}
Let $n>4$ and $(G,C_n)$ be a pair satisfying $(P1-3)$. Then every edge of $C_n$ is in exactly one $BBR$ triangle.
\end{lemma}
\begin{proof}
Let $\rep$ be a representation of $(G,C_n)$ satisfying $(P3)$. As $(G,C_n)$ is not contractible, every edge of $C_n$ is in at least one $BBR$ triangle. Assume, by contradiction and without loss of generality, that the blue edge $\set{1,2}$ is part of the two possible $BBR$ triangles $\set{0,1,2}$ and $\set{1,2,3}$. $\set{0,3}$ is not an edge of $C_n$, because $n > 4$. Moreover, it is not an edge of $G$, because otherwise the sub-pair induced by $\set{0,1,2,3}$ is isomorphic to a $(K_4,P_4)$. Let $e_0$ (resp. $e_3$) be an edge of $T$ defining the edge-clique $\set{0,1,2}$ (resp. $\set{1,2,3}$) such that $e_3$ is closest to $e_0$. Clearly, $e_0 \neq e_3$, because otherwise we get a $(K_4,P_4)$. The vertices $\set{4,\ldots,n-1}$ constitute a connected component of $G$, therefore the union of the corresponding paths is a subtree $T'$ of $T$. $T'$ edge-intersects both $P_0$ and $P_3$, therefore there is at least one path $P_j \notin \set{P_0,P_1,P_2,P_3}$ that contains $e_3$. We conclude that $\set{1,2,3,j}$ is an edge-clique. If $j=4$ then it induces a pair isomorphic to $(K_4,P_4)$, otherwise $\set{1,3,j}$ is a red edge-clique. Both cases contradict our assumptions.
\end{proof}

\begin{lemma}\label{lem:contraction-wdt}
Let $(G,C)$ be a pair satisfying $(P2),(P3)$ and let $\wdtg$ be a weak dual tree of $(G,C)$. Then (i) there is a bijection between the contractible edges of $(G,C)$ and the intermediate vertices of $\wdtg$, (ii) the tree obtained from $\wdtg$ by smoothing out the intermediate vertex corresponding to a contractible edge $e$ is a weak dual tree of $\contractgce$.
\end{lemma}

\begin{proof}
(i) We define the bijection $f$ as follows: Let $\wdtg$ be the weak dual tree corresponding to some $\opg$, and let $e$ be a contractible edge of $(G,C)$. Then $e$ is not part of any $BBR$ triangle. As every blue edge must be in some triangle, $e$ is in a non-empty set of $BRR$ triangles. Exactly one of these triangles is in $\opg$, and this triangle corresponds to an intermediate vertex of $\wdtg$ that we designate as $f(e)$. $f$ is one-to-one because every intermediate vertex corresponds to one $BRR$ triangle of $\opg$, and every $BRR$ has one blue edge. We now show that $f$ is onto. Assume by contradiction that $f$ is not onto. Then, without loss of generality there is a $BRR$ triangle $\set{1,2,j}$ ($j \notin \set{0,1,2,3}$) of $\opg$ where $e=\set{1,2}$ is not contractible. Then either $\set{0,2}$ or $\set{1,3}$ is an edge of $E(G) \setminus E(C)$. Let, without loss of generality $\set{0,2}$ be an edge of $E(G) \setminus E(C)$. Then $\set{0,1,2}$ is an edge-clique. Let $E'$ be the set of edges of (the path of) $T$ defining this edge-clique. We claim that $\forall k \notin \set{0,1,2}, P_k \cap E' = \emptyset$. Indeed, if $k=n-1$ and $P_k$ contains an edge of $E'$, then $\set{n-1,0,1,2}$ induces a $(K_4,P_4)$, and if $k \neq n-1$ then $\set{k,0,2}$ induces a red edge-clique. In both cases we reach a contradiction. Consider the subtrees of $T$ separated by $E'$. As $P_j \cap E'=\emptyset$ it is completely contained in one of these subtrees, say $T_j$. $P_1$ and $P_2$ edge-intersect $P_j$, therefore they edge-intersect $T_j$. However $P_0$ and $T_j$ are edge-disjoint as this would either contradict the definition of $E'$ or $P_0$ would split from $P_1$. On the other hand, the vertices $\set{j+1,j+2 \ldots, 0}$ constitute a connected component of $G$, therefore the union of the paths $\set{P_{j+1},P_{j+2} \ldots, P_0}$ is a subtree $T'$ of $T$. $T'$ edge-intersects both $P_0$ and $P_j$, therefore $T'$ edge-intersects $E'$. In other words there is at least one path $P_l \in \set{P_{j+1},P_{j+2} \ldots, P_0}$ that edge-intersects $E'$, a contradiction.

(ii) Let $e=\set{i,i+1}$ and $\set{i,i+1,j}$ the $BRR$ triangle of $\opg$ (that corresponds to $f(e)$). After the contraction of $e$, this triangle reduces to a red edge. The same holds for $\contract{\opg}{e}$ that contains all the faces of $\opg$ except the $BRR$ triangle that disappeared. The corresponding weak dual tree is $\wdtg$ with $f(e)$ smoothed out.
\end{proof}

We note that if $n=4$ Lemma \ref{lem:BBRTrianglesDontOverlap} does not hold. However the following corollary of lemmata \ref{lem:BBRTrianglesDontOverlap} and \ref{lem:contraction-wdt} holds for every $n$.

\begin{coro}\label{coro:NoIntermediate}
If $(G,C)$ is a pair satisfying $(P1-3)$ with $C$ isomorphic to $C_n$, then: (i) $\wdtg$ does not have intermediate vertices, (ii) $n$ is even and $\wdtg$ has $n/2$ leaves, and (iii) $\wdtg$ is a path if and only if $n=4$.
\end{coro}

\subsection{The Minimal Representation}\label{subsec:noIntermediateVertices}
In this section we present an algorithm solving $\prbpthree$ for $n \geq 5$, provided that assumptions $(P1)-(P2)$ hold. The representation returned by the algorithm is a planar tour. We show that it is the unique minimal representation of $(G,C)$ satisfying $(P3)$.

\newcommand{\algone}{\textsc{BuildPlanarTour}}

\begin{lemma}\label{lem:P123ImpliesOuterPlanarAndPlanarTour}
If $(G,C)$ is a hamiltonian pair with $n=\abs{V(G)} > 4$ for which properties $(P1-3)$ hold then
$G$ is outerplanar and the unique minimal representation of $(G,C)$ satisfying $(P3)$ is a planar tour of the weak dual tree of $G$.
\end{lemma}
\begin{proof}
The proof is by induction on the smallest number $h$ of junctions of a weak dual tree of $(G,C)$. Let $\wdtg$ be a weak dual tree of $(G,C)$ with $h$ junctions, and $\opg$ the corresponding maximal outerplanar graph. Index arithmetic is modulo $n$ through the proof, and $\rep$ is a minimal representation of $(G,C)$ satisfying $(P3)$. We first recall that since $(G,C)$ satisfies $(P1)$, by Corollary \ref{coro:NoIntermediate}, $\wdtg$ contains only junctions and leaves. In the sequel we show that $T$ is isomorphic to $\wdtg$ and $\pp$ is a planar tour of $T$. We do this by combining planar tours of subtrees into a planar tour a tree. Two basic tools that we use in the construction are the following two claims that state, roughly speaking, a) that two adjacent holes of $\opg$ are represented by two pies with distinct centers, and b) that the representations associated with disjoint subtrees of $\wdtg$, reside in disjoint subtrees of $T$.

For a junction $x$ of $\wdtg$, let $H_x$ be the set of vertices of the hole corresponding to $x$ in $\opg$. By Property $(P3)$, $H_x$ is represented by a pie. We denote by $f(x)$ be the center of the pie in $\rep$ representing $H_x$.

\begin{claim}\label{claim:AdjacentHolesDistinctCenters}
If $u,v$ are two adjacent junctions of $\wdtg$ then $f(u) \neq f(v)$.
\end{claim}
\begin{proof}
Let $\set{i,j}$ be the edge common to the holes $H_u$ and $H_v$. In other words $\set{i,j}$ is the dual of the edge $\set{u,v}$ of $\wdtg$. Let $k \neq i$ and $k' \neq i$ be the vertices adjacent to $j$ in these two holes. Assume for a contradiction that $f(u)=f(v)$ (see Figure \ref{fig:AdjacentHoles} for an illustration).
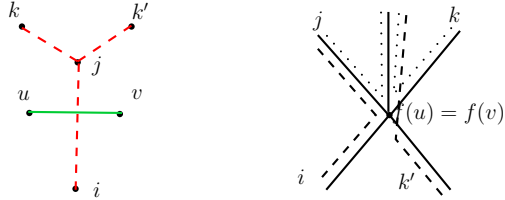
\begin{figure}[htbp]
\centering
\scalebox{0.7} 
{
\begin{pspicture}(0,-1.8929688)(9.342813,1.8929688)
\definecolor{color299}{rgb}{0.0,0.8,0.2}
\psdots[dotsize=0.12](0.5209375,-0.26546875)
\psdots[dotsize=0.12](2.2409375,-0.28546876)
\psdots[dotsize=0.12](0.3809375,1.3745313)
\psdots[dotsize=0.12](2.4609375,1.3745313)
\psdots[dotsize=0.12](1.4409375,0.69453126)
\psdots[dotsize=0.12](1.4009376,-1.7054688)
\psline[linewidth=0.04cm,linecolor=red,linestyle=dashed,dash=0.16cm 0.16cm](0.3809375,1.3545313)(1.4409375,0.71453124)
\psline[linewidth=0.04cm,linecolor=red,linestyle=dashed,dash=0.16cm 0.16cm](1.4409375,0.71453124)(2.4609375,1.3945312)
\psline[linewidth=0.04cm,linecolor=red,linestyle=dashed,dash=0.16cm 0.16cm](1.4609375,0.6745312)(1.4009376,-1.7254688)
\psline[linewidth=0.04cm,linecolor=color299](0.5009375,-0.24546875)(2.2809374,-0.26546875)
\usefont{T1}{ptm}{m}{n}
\rput(1.8023437,-1.7154688){$i$}
\usefont{T1}{ptm}{m}{n}
\rput(1.8123437,0.62453127){$j$}
\usefont{T1}{ptm}{m}{n}
\rput(0.23234375,1.7045312){$k$}
\usefont{T1}{ptm}{m}{n}
\rput(2.6323438,1.6845312){$k'$}
\usefont{T1}{ptm}{m}{n}
\rput(0.41234374,0.06453125){$u$}
\usefont{T1}{ptm}{m}{n}
\rput(2.5723438,0.08453125){$v$}
\psdots[dotsize=0.12](7.3609376,-0.30546874)
\psline[linewidth=0.04cm](6.0009375,1.2145313)(7.3809376,-0.34546876)
\psline[linewidth=0.04cm](7.3409376,-0.24546875)(7.3409376,1.6545312)
\psline[linewidth=0.04cm](7.3609376,-0.32546875)(8.700937,1.2945312)
\psline[linewidth=0.04cm](7.3609376,-0.30546874)(8.520938,-1.7254688)
\psline[linewidth=0.04cm](7.3609376,-0.30546874)(6.1609373,-1.7254688)
\psline[linewidth=0.04,linestyle=dotted,dotsep=0.16cm](6.1609373,1.2145313)(7.2209377,0.03453125)(7.2009373,1.6545312)
\psline[linewidth=0.04,linestyle=dotted,dotsep=0.16cm](7.5409374,1.6345313)(7.4409375,0.09453125)(8.540937,1.3545313)
\psline[linewidth=0.04,linestyle=dashed,dash=0.16cm 0.16cm](7.6809373,1.5945313)(7.4809375,-0.7654688)(8.360937,-1.8254688)
\psline[linewidth=0.04,linestyle=dashed,dash=0.16cm 0.16cm](6.0009375,0.95453125)(7.1009374,-0.30546874)(6.0809374,-1.4854687)
\usefont{T1}{ptm}{m}{n}
\rput(5.682344,-1.4754688){$i$}
\usefont{T1}{ptm}{m}{n}
\rput(7.6923437,-1.5554688){$k'$}
\usefont{T1}{ptm}{m}{n}
\rput(8.592343,1.6045313){$k$}
\usefont{T1}{ptm}{m}{n}
\rput(6.012344,1.4845313){$j$}
\usefont{T1}{ptm}{m}{n}
\rput(8.482344,-0.27546874){$f(u)=f(v)$}
\end{pspicture} 
}
\caption{Adjacent holes of $\opg$ are mapped to different centers.}\label{fig:AdjacentHoles}
\end{figure}
$P_i,P_j,P_k$ are consecutive in one pie and $P_i,P_j,P_{k'}$ are consecutive in the other. Then $P_k$ and $P_{k'}$ edge-intersect $P_j$ on the same edge  (incident to $f(u)$), thus forming an edge-clique of $G$. We will show that this is a red edge-clique of $G$, contradicting $(P3)$. Indeed, $\set{j,k}$ and $\set{j,k'}$ are red edges of $\opg$. If $\set{k,k'}$ is a blue edge then $\set{j,k,k'}$ constitutes a $BRR$ triangle of $\opg$ that corresponds to an intermediate vertex of $\wdtg$, contradicting Corollary \ref{coro:NoIntermediate}.
\end{proof}

\begin{claim}\label{claim:SubtreesOfTheRepresentationsAreIndependent}
Let $u$ be a junction of degree $d$ of a weak dual tree of $(G,C)$. Let $S_1,S_2,\ldots,S_d$ be the connected components of $C \setminus H_u$, and let $\pp_i$ be the set of paths representing the vertices of $S_i$ in a minimal representation. Then $H_u$ is represented by a pie with edges $e_1,e_2,\ldots,e_d$ whose removal together with $f(u)$ divides $T$ into subtrees $T_1, T_2, \ldots, T_d$, such that:
\begin{enumerate}[i)]
\item \label{itm:AllAreInTiEi} $\cup \pp_i \subseteq T_i + e_i$ for every $i \in [d]$,
\item \label{itm:NonSingleton} $\cup \pp_i \subseteq T_i$ whenever $S_i$ is not a singleton, and
\item \label{itm:Singleton} $E(T_i) = \emptyset$ whenever $S_i$ is a singleton.
\end{enumerate}
\end{claim}
\begin{proof}
\ref{itm:AllAreInTiEi},\ref{itm:NonSingleton}) The removal of $f(u)$ from $T$ (together with its incident edges) defines at least $d$ subtrees $T_1, T_2, \ldots T_d$ of $T$ where $e_i$ has one endpoint in $T_i$ for $i \in [d]$. We consider two vertices $i,j$ consecutive on the hole $H_u$. Without loss of generality $P_i$ contains $e_0$ and $e_1$, $P_j$ contains $e_1$ and $e_2$. Consider the \emph{segment} (i.e. connected component) $S=\set{i+1,i+2,\ldots,j-1}$ of $G \setminus H_u$. We will conclude the proof of \ref{itm:AllAreInTiEi},\ref{itm:NonSingleton}) by showing that $\cup \pp_S \subseteq T_1$ where $\pp_S$ is the set of paths in $\pp$ that represent the vertices of $S$.

If $S$ contains at least two vertices, then the hole adjacent to $H_u$ is a red hole $H_v$. By Claim \ref{claim:AdjacentHolesDistinctCenters}, $f(u) \neq f(v)$. Since $f(u),f(v) \in \split(P_i,P_j)$ and $\abs{\split(P_i,P_j)} \leq 2$ this implies that $\split(P_i,P_j)=\set{f(u),f(v)}$.

Let $P=p_T(f(u),f(v))$ and let $T_u,T_v,T'_1,T'_2,\ldots$ be the trees of the forest obtained by the removal of the edges of $P$ from $T$, where $V(T_u) \cap V(P)=f(u)$, $V(T_v) \cap V(P)=f(v)$ and $V(T'_\ell) \cap V(P)$ is an intermediate vertex of $P$. We observe that if a path $P_k$ edge-intersects some subtree $T'_\ell$ in at least one edge and also edge-intersects $P$, then $\set{i,j,k}$ is a red edge-clique, contradicting property $(P3)$. Therefore, every path edge-intersecting $T'_\ell$ is contained in $T'_\ell$ implying that set of vertices represented by such paths are disconnected from the rest of $G$, contradicting the connectedness of $G$. We conclude that the tree $T'_\ell$ contains no paths and since $\rep$ is minimal, $T'_\ell$ consists of a single vertex, namely an intermediate vertex of $P$. Summarizing, we have $T = T_u \cup T_v \cup P$. Finally, we note that $T_1 = P \cup T_v$. In the sequel we show that $\pp_S \subseteq T_v$ and $P$ consists of the edge $e_1$.

Let $S' = S \setminus \set {i+1,j-1}$. $H_v$ contains at least one vertex $k \in S'$, and $P_k$ is part of the pie centered at $f(v)$. Therefore, $P_k \subseteq T_v$, implying that $\pp_{S'} \cap T_v \neq \emptyset$. If $P_{k'}$ crosses $v$ for some $k' \in S'$ then $\set{i,j,k'}$ constitutes a red edge clique, contradicting property $(P3)$. Therefore, $\cup \pp_{S'} \subseteq T_v$. We now show that $P_{i+1} \cup P_{j-1} \subseteq T_v$. Since $i+1$ (resp. $j-1$) is adjacent to $i+2 \in S'$ (resp. $j-2 \in S'$), both of $P_{i+1}$ and $P_{j-2}$ edge-intersect $\cup \pp_{S'}$ implying that the both edge-intersect $T_v$. We now consider the $BBR$ triangle $j,j+1,j+2$. We have $P_j \nsim P_{j+2}$, therefore $\emptyset \neq \split(P_j,P_{j+2}) \subseteq V(P_{j+2}) \subseteq V(T_v)$. Let $x$ be vertex of $\split(P_j,P_{j+2})$ closest to $v$ (possibly $x=v$). Assume, by way of contradiction that $P_{j+1}$ crosses $v$. If $P_{j+1}$ does not cross $x$ then $P_{j+1} \parallel P_{j+2}$, otherwise $P_{j+1} \nsim P_{j+2}$ or $P_{j+1} \nsim P_j$. Both cases contradict the fact that $j,j+1,j+2$ are consecutive in $C$. Therefore, $P_{j+1}$ does not cross $v$, i.e. $P_{j+1} \subseteq T_v$. Similarly, $P_{i-1} \subseteq T_v$. We conclude that $\cup \pp_S \subseteq T_v$. Since the only paths edge-intersecting $P$ are $P_i$ and $P_j$, and by the minimality of the representation, $P$ consists of only one edge, namely $e_1$, concluding the proof of \ref{itm:AllAreInTiEi}) and \ref{itm:NonSingleton}) for this case. Otherwise, $S$ is a singleton. Then $S=\set{i+1}$ and $i,i+1,j$ are consecutive in $C$. Therefore, $P_{i+1} \sim P_i$ and $P_{i+1} \sim P_j$. Then, $P_{i+1}$ is contained in $T_1+e_1$

\ref{itm:Singleton}) If $S$ is a singleton, the only paths edge-intersecting $T_1+e_1$ are $P_i,P_{i+1},P_j$, since all the other paths are in their respective subtrees, each disjoint from $T_1$. Together with the minimality of the representation, this implies that $P_{i+1}$ consists of the single edge $e_1$ and $E(T_1)=\emptyset$.
\end{proof}

We now proceed with the proof the lemma. If $h=0$, $\wdtg$ contains at most two vertices implying that $n \leq 4$. Therefore, $h \geq 1$.

Consult Figure \ref{fig:inductive-proof} (a) and (b) for the following discussion. Let $u$ be a junction of $\wdtg$, $H_u=\set{h_0,h_1,\ldots,h_{d-1}}$, $S_i$ be the segment of $C \setminus H_u$ between $h_i$ and $h_{i+1}$, and let $e_1,\ldots,e_d$, $T_1,\ldots,T_d$ as in Claim \ref{claim:SubtreesOfTheRepresentationsAreIndependent}.

If $h=1$ all the segments $S_i, i \in[0,d-1]$ are singletons. Then $T_i=\emptyset$ and the only vertex $h_i+1$ of $S_i$ is represented by a path consisting of $e_i$, for every $i \in [0,d-1]$. Therefore, $T$ is a star isomorphic to $\wdtg$ and $\pp$ is a planar tour of it.

If $h > 1$ for a singleton segment $S_i$ we have $T_i=\emptyset$ and the only vertex $h_i+1$ of $S_i$ is represented by a path consisting of $e_i$. For a segment $S_i$ consisting of at least two vertices we proceed as follows: the edge $h_i,h_{i+1}$ separates $H_u$ from another red hole $H_v$ where $f(u) \neq f(v)$ by Claim \ref{claim:AdjacentHolesDistinctCenters}. This implies that $\split(P_{h_i},P_{h_{i+1}})=\set{f(u),f(v)}$, and without loss of generality $f(v) \in V(T_i)$.

For the following discussion see Figure \ref{fig:inductive-proof} (c) and (d). Let $\bar{S}_i = S_i \cup \set{h_i,h_{i+1}}$ and let $(\bar{G}_i,\bar{C}_i)$ be the pair obtained from the pair $(G[\bar{S}_i],C[\bar{S}_i])$ by adding to it a new vertex $v_i$ and two edges $\set{v_i,h_i},\set{v_i,h_{i+1}}$. Let also $\bar{T}_i = T_i + e_i$ and $\bar{\pp}_i= \pp_{S_i} \cup \set{P_{h_i} \cap \bar{T}_i, P_{h_{i+1}} \cap \bar{T}_i,P_{v_i}}$ where $P_x$ is the path consisting of the edge $e_i$. Then $\reptp{\bar{T}_i}{\bar{\pp}_i}$ is a representation of  $(\bar{G}_i,\bar{C}_i)$, since the paths $P_{h_i}$ and $P_{h_{i+1}}$ split in a vertex $f(v)$ of $T_i$ and the other parts are completely contained in $T_i$, by Claim \ref{claim:SubtreesOfTheRepresentationsAreIndependent}.

We note that a minifying operation applied to an edge of $T_i$ to get a representation equivalent to $\reptp{\bar{T}_i}{\bar{\pp}_i}$ can be applied to $\rep$ to get an equivalent representation to $\rep$ contradicting the minimality of $\rep$. Moreover, any minifying operation on $e_i$ will cause $P_{v_i}$ to have an empty edge-intersection with $P_{h_i}$ or with $P_{h_{i+1}}$. Therefore, $\reptp{\bar{T}_i}{\bar{\pp}_i}$ is minimal. By the inductive hypothesis a) $\bar{G}_i$ is outerplanar, b) $\bar{\pp}_i$ is a planar tour of $\bar{T}_i$, and, c) $\bar{T_i}$ is isomorphic to the weak dual tree of $\bar{G}_i$. Consider $\wdtg$ as rooted at $u$, and let $\wdt_i$ be the subtree of $u$ containing $v$. We note that the weak dual tree of $\bar{G}_i$ is isomorphic to $\wdt_i$.

It remains to observe that $T$ is the union of all trees $\bar{T}_i$, i.e. isomorphic to $\wdtg$, and that $\pp$ is a planar tour of it. $G$ is outerplanar, since each $G_i$ is outerplanar, and if $G$ is not outerplanar, then there is an edge between a vertex of $S_i$ and a vertex of $S_j$ for $i,j \in [d]$ and $i \neq j$. However, this is a contradiction to the fact that, by Claim \ref{claim:SubtreesOfTheRepresentationsAreIndependent}, $\cup S_i \subseteq \bar{T}_i$, $\cup S_j \subseteq \bar{T}_j$ and $\bar{T}_i, \bar{T}_j$ are edge-disjoint.
\begin{figure}[htbp]
\centering
\include{figure/inductive-proof}
\caption{(a), (b), (c), (d) An induction step of the proof of Lemma \ref{lem:P123ImpliesOuterPlanarAndPlanarTour} illustrated for $d = 5$. To keep the figure simple, most of the paths are omitted and the segments having more than one vertex, i.e. $S_1, S_3, S_5$, are depicted by arcs. (e) The unique minimal representation of $(G,C)$ satisfying $(P3)$ obtained by combining the subtrees and the paths of the segments with a representation of a hole $H_u$.}
\label{fig:inductive-proof}
\end{figure}

\end{proof}

\begin{lemma}\label{lem:OuterPlanarImpliesP13}
If $(G,C)$ is a hamiltonian pair on $n>4$ vertices and $G$ is outerplanar such that every edge of the outer face of $G$ is contained in a $BBR$ triangle, then properties $(P1-3)$ hold for $(G,C)$.
\end{lemma}
\begin{proof}
$(G,C_n)$ satisfies $(P1)$ since every edge of $C_n$ is in a $BBR$ triangle.
Since $G$ is outerplanar it does not contain a $K_4$. Therefore, $(G,C_n)$ satisfies $(P2)$. To prove that $(P3)$ holds, we show by induction on the number $h$ of junctions of the weak dual tree $\wdtg$ of $G$ that a planar tour of $\wdtg$ is a representation of $(G,C_n)$ satisfying $(P3)$:

If $h=1$ then $\wdtg$ is a star, and $G=\opg$ is a hole surrounded by $BBR$ triangles. Then, a planar tour of $\wdtg$ is a representation of $(G,C_n)$ satisfying $(P3)$.

If $h>1$ we pick an chord $e=\set{i,j}$ of $C$ separating two red holes of $G$ and construct two pairs $(G',C'), (G'',C'')$ in a way similar to the proof of Lemma \ref{lem:P123ImpliesOuterPlanarAndPlanarTour}. $V'=\set{i,i+1,\ldots,j}$ and $V''=\set{j,j+1,\ldots,i}$. $(G',V')$ (resp. $(G'',V'')$) consists of $(G[V'],C[V'])$ (resp. $(G[V''],C[V''])$) and an additional vertex with two adjacent edges closing the cycle. By the inductive assumption, a planar tour of $\wdtgprime$ (resp. $\wdtgdprime$) is a representation of $(G',C')$ (resp. $(G'',C'')$). Removing from these planar tours the short paths corresponding to the vertices not in $V(G)$ and gluing together the rest by identifying the common endpoints of the paths $P_i, P_j$ we get a planar tour of $\wdtg$ that represents $(G,C)$.
\end{proof}

We get the following theorem as a corollary of lemmata \ref{lem:BBRTrianglesDontOverlap}, \ref{lem:P123ImpliesOuterPlanarAndPlanarTour} and \ref{lem:OuterPlanarImpliesP13}.

\begin{theorem}\label{theo:P1-3EquivalentToOuterPlanar}
The following statements are equivalent whenever $n>4$:
\begin{enumerate}[i)]
\item {} $(G,C_n)$ satisfies assumptions $(P1-3)$.
\item {} $(G,C_n)$ has a unique minimal representation satisfying $(P3)$ which is a planar tour of a weak dual tree of $G$.
\item {} $G$ is Hamiltonian outerplanar and every face adjacent to the unbounded face $F$ is a triangle having two edges in common with $F$, (i.e. a BBR triangle).
\end{enumerate}
\end{theorem}

Algorithm $\algone$ calculates a planar tour of a weak dual tree $\wdtg$. Therefore,

\begin{theorem}\label{thm:BuildPlanarTourCorrect}
Instances of $\prbpthree$ satisfying properties $(P1),(P2)$ can be solved in polynomial time.
\end{theorem}

\alglanguage{pseudocode}

\begin{algorithm}
\footnotesize
\caption{$\algone(G,C)$}
\label{alg:algone}
\begin{algorithmic}[1]
\Require {$\abs{V(G)} \geq 5$}
\Require {$(G,C)$ satisfies assumptions $(P1),(P2)$}
\Ensure {$\repbar$ is the unique minimal representation of $(G,C)$ satisfying $(P3)$}
\If {$G$ is not outerplanar}
\State \Return ``NO''
\EndIf
\State $\bar{T} \leftarrow \wdtg$.
\Comment Corresponding to $\opg$
\Statex \textbf{Build the planar tour:}
\State Let $\set{v_0, v_1, \ldots, v_{k-1}}$ be the leaves of $\bar{T}$ ordered as they are
\State encountered in a DFS traversal of $\bar{T}$ corresponding to the planar embedding
\State suggested by $\opg$.
\State Let $L_i = p_{\bar{T}}(v_i,v_{(i+1) \mod k})$
\State Let $S_i$ be the path of length $1$ starting at $v_i$.
\State $\ppbar_L \gets \set{L_i|~0 \leq i \leq n-1}$.
\State $\ppbar_S \gets \set{S_i|~0 \leq i \leq n-1}$.
\State Let $\bar{P}_i= \left\{ \begin{array}{ll}
L_{i/2} & \textrm{if~}i\textrm{~is even}\\
S_{\lfloor i/2 \rfloor} & \textrm{otherwise}
\end{array} \right.$
\State $\ppbar \leftarrow \set{\bar{P}_i|~ 0 \leq i \leq 2n-1}$
\Comment $=\ppbar_L \cup \ppbar_S$
\State \Return $\repbar$
\end{algorithmic}
\end{algorithm}

Figure \ref{fig:P1P3-cycle-representation} depicts a YES instance of $\prbpthree$.

\begin{figure}[htbp]
\centering
\scalebox{0.9} 
{
\begin{pspicture}(0,-3.2517188)(17.097187,3.2517188)
\definecolor{color530}{rgb}{0.0,0.0,0.8}
\definecolor{color540}{rgb}{1.0,0.0,0.2}
\psline[linewidth=0.04cm](0.9609375,2.2132812)(2.6409376,-0.08671875)
\psline[linewidth=0.04cm](2.6409376,-0.08671875)(0.4009375,-2.3867188)
\psline[linewidth=0.04cm](2.6209376,-0.10671875)(6.0009375,-0.10671875)
\psline[linewidth=0.04cm](6.0009375,-0.10671875)(7.5809374,1.2532812)
\psline[linewidth=0.04cm](7.5809374,1.2532812)(7.5809374,2.6132812)
\psline[linewidth=0.04cm](7.5809374,1.2732812)(9.580937,0.81328124)
\psline[linewidth=0.04cm](6.0009375,-0.10671875)(7.3609376,-2.5067186)
\psline[linewidth=0.04,linestyle=dashed,dash=0.16cm 0.16cm,dotsize=0.07055555cm 2.0]{**-**}(0.2009375,-2.1467187)(2.3209374,-0.047480654)(1.1409374,1.6132812)(0.75666565,2.1332812)
\psline[linewidth=0.04cm,fillcolor=black,linestyle=dashed,dash=0.16cm 0.16cm,dotsize=0.07055555cm 2.0]{**-**}(2.7609375,0.19328125)(1.1809375,2.2732813)
\psline[linewidth=0.04cm,fillcolor=black,linestyle=dashed,dash=0.16cm 0.16cm,dotsize=0.07055555cm 2.0]{**-**}(7.7209377,1.0732813)(9.520938,0.61328125)
\psline[linewidth=0.04cm,fillcolor=black,linestyle=dashed,dash=0.16cm 0.16cm,dotsize=0.07055555cm 2.0]{**-**}(6.3809376,-0.24671875)(7.6209373,-2.4467187)
\psline[linewidth=0.04cm,fillcolor=black,linestyle=dashed,dash=0.16cm 0.16cm,dotsize=0.07055555cm 2.0]{**-**}(2.7609375,-0.32671875)(0.7209375,-2.5667188)
\usefont{T1}{ptm}{m}{n}
\rput(8.132343,3.0632813){$P_{2.3}$}
\usefont{T1}{ptm}{m}{n}
\rput(10.052343,0.60328126){$P_5$}
\usefont{T1}{ptm}{m}{n}
\rput(7.6723437,0.08328125){$P_6$}
\usefont{T1}{ptm}{m}{n}
\rput(7.972344,-2.6167188){$P_7$}
\usefont{T1}{ptm}{m}{n}
\rput(5.7323437,-1.7767187){$P_8$}
\usefont{T1}{ptm}{m}{n}
\rput(0.49234375,-2.8167188){$P_9$}
\usefont{T1}{ptm}{m}{n}
\rput(1.3223437,-0.09671875){$P_{10}$}
\psline[linewidth=0.04,linestyle=dashed,dash=0.16cm 0.16cm,dotsize=0.07055555cm 2.0]{**-**}(9.440937,0.33328125)(7.5809374,0.8332813)(6.7009373,-0.08671875)(7.9609375,-2.2467186)
\usefont{T1}{ptm}{m}{n}
\rput(4.3323436,0.72328126){$P_1$}
\psdots[dotsize=0.12](0.9609375,2.1932812)
\psdots[dotsize=0.12](2.6209376,-0.08671875)
\psdots[dotsize=0.12](0.4409375,-2.3667188)
\psdots[dotsize=0.12](7.3409376,-2.4667187)
\psdots[dotsize=0.12](6.0009375,-0.10671875)
\psdots[dotsize=0.12](7.5809374,1.2732812)
\psdots[dotsize=0.12](9.580937,0.81328124)
\usefont{T1}{ptm}{m}{n}
\rput(0.70234376,2.5432813){$P_{11}$}
\psdots[dotsize=0.12](7.5809374,2.5932813)
\psline[linewidth=0.04cm](7.6009374,1.2732812)(8.340938,1.0932813)
\psline[linewidth=0.04,linestyle=dashed,dash=0.16cm 0.16cm,dotsize=0.07055555cm 2.0]{**-**}(1.4809375,2.4132812)(2.9809375,0.29328126)(5.9209375,0.25328124)(7.3409376,1.4532813)(7.3209376,2.6332812)
\psline[linewidth=0.04,linestyle=dashed,dash=0.16cm 0.16cm,dotsize=0.07055555cm 2.0]{**-**}(8.020938,2.5932813)(8.060938,1.4132812)(9.880938,1.0332812)
\psline[linewidth=0.04cm,fillcolor=black,linestyle=dashed,dash=0.16cm 0.16cm,dotsize=0.07055555cm 2.0]{**-**}(7.8009377,2.6532812)(7.8009377,1.3932812)
\usefont{T1}{ptm}{m}{n}
\rput(8.612344,1.7832812){$P_4$}
\psline[linewidth=0.04,linestyle=dashed,dash=0.16cm 0.16cm,dotsize=0.07055555cm 2.0]{**-**}(1.1009375,-2.6067188)(3.0209374,-0.44671875)(5.7809377,-0.40671876)(7.1209373,-2.6467187)
\psdots[dotsize=0.12](14.300938,2.4132812)
\psdots[dotsize=0.12](14.300938,-2.7667189)
\psdots[dotsize=0.12](11.800938,0.53328127)
\psdots[dotsize=0.12](16.720938,-1.0867188)
\psdots[dotsize=0.12](12.820937,1.9732813)
\psdots[dotsize=0.12](15.780937,-2.3067188)
\psdots[dotsize=0.12](15.720938,1.9932812)
\psdots[dotsize=0.12](12.900937,-2.3467188)
\psdots[dotsize=0.12](16.700937,0.73328125)
\psdots[dotsize=0.12](11.940937,-1.2867187)
\psline[linewidth=0.04cm,linecolor=color530](11.800938,0.5132812)(11.940937,-1.2867187)
\psline[linewidth=0.04cm,linecolor=color530](11.940937,-1.2867187)(12.920938,-2.3667188)
\psline[linewidth=0.04cm,linecolor=color530](12.920938,-2.3667188)(14.320937,-2.7667189)
\psline[linewidth=0.04cm,linecolor=color530](14.320937,-2.7667189)(15.780937,-2.3067188)
\psline[linewidth=0.04cm,linecolor=color530](15.780937,-2.3067188)(16.720938,-1.1067188)
\psline[linewidth=0.04cm,linecolor=color530](16.720938,-1.1067188)(16.700937,0.75328124)
\psline[linewidth=0.04cm,linecolor=color530](16.700937,0.75328124)(15.700937,2.0132813)
\psline[linewidth=0.04cm,linecolor=color530](15.700937,2.0132813)(14.300938,2.4132812)
\psline[linewidth=0.04cm,linecolor=color530](14.300938,2.4132812)(12.840938,1.9932812)
\psline[linewidth=0.04cm,linecolor=color530](12.840938,1.9932812)(11.800938,0.5132812)
\psline[linewidth=0.04cm,linecolor=color540,linestyle=dashed,dash=0.16cm 0.16cm](14.280937,2.4332812)(16.680937,0.75328124)
\usefont{T1}{ptm}{m}{n}
\rput(14.267813,2.7232811){1}
\usefont{T1}{ptm}{m}{n}
\rput(16.012032,2.2432814){2.3}
\usefont{T1}{ptm}{m}{n}
\rput(16.901875,0.94328123){4}
\usefont{T1}{ptm}{m}{n}
\rput(16.970469,-1.0367187){5}
\usefont{T1}{ptm}{m}{n}
\rput(16.05625,-2.6967187){6}
\usefont{T1}{ptm}{m}{n}
\rput(14.39625,-3.0967188){7}
\usefont{T1}{ptm}{m}{n}
\rput(12.729688,-2.5567188){8}
\usefont{T1}{ptm}{m}{n}
\rput(11.536094,-1.2967187){9}
\usefont{T1}{ptm}{m}{n}
\rput(11.393437,0.5232813){10}
\usefont{T1}{ptm}{m}{n}
\rput(12.517813,2.2232811){11}
\psline[linewidth=0.04cm,linecolor=color540,linestyle=dashed,dash=0.16cm 0.16cm](16.660938,0.7132813)(15.800938,-2.2867188)
\psline[linewidth=0.04cm,linecolor=color540,linestyle=dashed,dash=0.16cm 0.16cm](15.780937,-2.2867188)(12.900937,-2.3267188)
\psline[linewidth=0.04cm,linecolor=color540,linestyle=dashed,dash=0.16cm 0.16cm](12.900937,-2.3067188)(11.780937,0.55328125)
\psline[linewidth=0.04cm,linecolor=color540,linestyle=dashed,dash=0.16cm 0.16cm](11.780937,0.55328125)(14.320937,2.3932812)
\psline[linewidth=0.04cm,linecolor=color540,linestyle=dashed,dash=0.16cm 0.16cm](14.280937,2.3932812)(12.900937,-2.3867188)
\psline[linewidth=0.04cm,linecolor=color540,linestyle=dashed,dash=0.16cm 0.16cm](14.320937,2.4332812)(15.780937,-2.2867188)
\psline[linewidth=0.04cm,linecolor=green](12.900937,1.7932812)(12.980938,0.17328125)
\psline[linewidth=0.04cm,linecolor=green](12.980938,0.17328125)(12.120937,-1.1067188)
\psline[linewidth=0.04cm,linecolor=green](12.980938,0.15328126)(14.340938,-0.48671874)
\psline[linewidth=0.04cm,linecolor=green](14.340938,-0.48671874)(15.780937,0.21328124)
\psline[linewidth=0.04cm,linecolor=green](15.780937,0.21328124)(15.620937,1.7932812)
\psline[linewidth=0.04cm,linecolor=green](15.780937,0.21328124)(16.420937,-0.7867187)
\psline[linewidth=0.04cm,linecolor=green](14.340938,-0.50671875)(14.260938,-2.5267189)
\psdots[dotsize=0.12,linecolor=green](12.980938,0.15328126)
\psdots[dotsize=0.12,linecolor=green](12.100938,-1.1267188)
\psdots[dotsize=0.12,linecolor=green](12.900937,1.7732812)
\psdots[dotsize=0.12,linecolor=green](14.340938,-0.50671875)
\psdots[dotsize=0.12,linecolor=green](14.260938,-2.5267189)
\psdots[dotsize=0.12,linecolor=green](16.440937,-0.7867187)
\psdots[dotsize=0.12,linecolor=green](15.640938,1.7732812)
\end{pspicture}
}
\caption{A pair $(G,C)$, its weak dual tree $\wdtg$ and the representation of $(G,C)$ returned by $\algone$.}
\label{fig:P1P3-cycle-representation}
\end{figure}
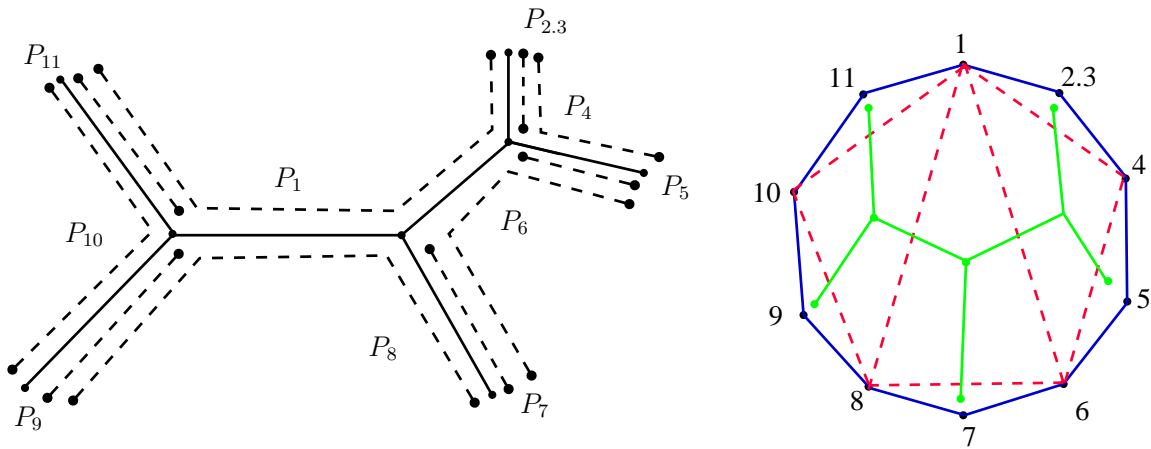

\section{Conclusion and Further Research}\label{sec:conclusion}
In this work we define the family of $\enpt$ graphs that are subgraphs of $\ept$ graphs. We showed that trees, cliques and holes are $\enpt$ graphs. We then started the study of characterization of representations of holes. It turns out that in $\enpt$ graphs, representations of holes have a much more complicated structure, yielding several possible representations. We considered $\ept$, $\enpt$ pairs of graphs, i.e. the $\ept$ and $\enpt$ graphs of a given representation. We defined three properties $(P1)$, $(P2)$, $(P3)$ on these pairs of graphs. We characterized the representations of pairs having all three properties as planar tours of the weak dual of the $\ept$ graph.

Our investigation in \cite{BESZ13-TR2} deals with the generalization of the results to representations that do not satisfy assumptions $(P1-3)$. We first relax assumption $(P1)$, and show that the unique minimal representation of a pair satisfying $(P2-3)$ is a \emph{broken planar tour}, which is obtained from a planar tour by \emph{breaking apart} individual paths into sub-paths. Then, we relax assumption $(P2)$, and show that the unique minimal representation of a pair satisfying $(P3)$ is a \emph{broken planar tour with cherries}, where a cherry is a sub-graph of a tree isomorphic to a $P_3$ whose leaves are leaves of the tree. For both cases we provide polynomial-time algorithms to find these representations. We also show that if assumption $(P3)$ is relaxed, such a representation can not be found in polynomial time, unless $\textsc{P}=\textsc{NP}$.

\end{document}